\documentclass[USenglish,letterpaper]{lipics}
\usepackage{xr-hyper}
\usepackage{hyperref}
\usepackage{mathpartir}
\usepackage{multicol}
\usepackage{amsmath}
\usepackage{amsthm}
\usepackage{amssymb}
\usepackage{stmaryrd}
\usepackage[numbers]{natbib}
\usepackage{rotating}
\usepackage{wasysym}
\usepackage{wrapfig}
\usepackage[finalnew]{trackchanges}
\usepackage{listings}
\usepackage{multicol}
\usepackage{wrapfig}
\usepackage[usenames,dvipsnames]{xcolor}
\usepackage{color}
\usepackage{tikz}
\usetikzlibrary{calc}
\usepackage{float}
\usepackage{textcomp}

\newcommand{\hosturl}{the ancilliary files for this arXiv submission}

\usetikzlibrary{arrows,positioning} 
\tikzset{
    >=stealth',
    punkt/.style={
           rectangle,
           rounded corners,
           draw=black, very thick,
           text width=6.5em,
           minimum height=2em,
           text centered},
    pil/.style={
           ->,
           thick,
           shorten <=2pt,
           shorten >=2pt,}
}

\newtheorem{mainlemma}{Main Lemma}
\newtheorem{techlemma}{Technical Lemma}

\newlength\Colsep
\definecolor{dkgreen}{rgb}{0,0.6,0}
\definecolor{gray}{rgb}{0.5,0.5,0.5}
\definecolor{lightgray}{rgb}{0.85,0.85,0.85}
\definecolor{mauve}{rgb}{0.58,0,0.82}
\definecolor{dkyellow}{rgb}{.8, .8, 0}

\newcommand{\transient}{\textsf{transient}}
\newcommand{\guarded}{\textsf{guarded}}
\newcommand{\monotonic}{\textsf{monotonic}}

\newcommand{\dyn}{\star}

\newcommand{\tcheck}[2]{#1\Downarrow_{#2}}
\newcommand{\lett}[3]{\texttt{let}~#1=#2~\texttt{in}~#3}
\newcommand{\classt}[8]{\texttt{class}(#5){:}\llparenthesis#2,#3,#4\rrparenthesis~\texttt{init}{=}#8; #6; #7} 
\newcommand{\funt}[3]{\lambda #1 {\to} #2.\;#3}
\newcommand{\method}[3]{\varsigma\;x_s, #1 {\to} #2.\;#3}
\newcommand{\constructt}[2]{\sigma\;x_s, #1.\;#2}
\newcommand{\classty}[5]{\textsf{class}~\llparenthesis#2,#3,#4,#5\rrparenthesis} 
\newcommand{\objty}[3]{\textsf{object}~\llparenthesis#2,#3\rrparenthesis} 
\newcommand{\bclassty}[2]{\textsf{class}~\llparenthesis#1,#2\rrparenthesis} 
\newcommand{\bobjty}[1]{\textsf{object}~\llparenthesis#1\rrparenthesis} 
\newcommand{\bdyn}{\textsf{pyobj}}
\newcommand{\bto}[1]{#1{\to}~}
\newcommand{\chole}{\square}
 
\newcommand{\classe}[5]{\texttt{class}(#2){:}^{#5}~\texttt{init}=#4; #3} 
\newcommand{\fune}[2]{\lambda #1.\;#2}
\newcommand{\classh}[3]{\texttt{Class}(#1)\{\texttt{init}=#3; #2\}} 
\newcommand{\objh}[2]{\texttt{Object}(#1)\{#2\}} 
\newcommand{\upython}{\textmu{}Python}
\newcommand{\cicolor}{dkyellow}
\newcommand{\pcolor}{dkyellow}

\newcommand{\flr}[1]{\lfloor #1\rfloor}

\newcommand{\cig}{\Gamma}
\newcommand{\tlg}{\Gamma{;}\Sigma}
\newcommand{\ninference}[3]{\inferrule[(#1)]{#2}{#3}}
\newcommand{\fninference}[3]{\inferrule*[left=(#1)]{#2}{#3}}
\newcommand{\inference}[2]{\inferrule{#1}{#2}}
\newcommand{\li}[1]{\lstinline{#1}}

\makeatletter
\providecommand{\sketchname}{Proof Sketch}
\newenvironment{sketch}[1][\sketchname]{\par
  \pushQED{\qed}%
  \normalfont \topsep6\p@\@plus6\p@\relax
  \trivlist
  \item[\hskip\labelsep
        \itshape
    #1\@addpunct{.}]\ignorespaces
}{%
  \popQED\endtrivlist\@endpefalse
}
\makeatother

\addeditor{Mike}
\addeditor{Jeremy}

\begin{document}

\title{Gradual Typing in an Open World}
% Gradual typing in an open world?
%\author{}
\Copyright{ }
%\authorinfo{Anonymous}{}{}
\author{Michael M. Vitousek}
\author{Jeremy G. Siek}
\affil{Indiana University Bloomington\\
\texttt{\{mvitouse,jsiek\}@indiana.edu}}
%s/allow/enable or cause
%s/in order to/to
\maketitle
\begin{abstract}
  Gradual typing combines static and dynamic typing in the same
  language, offering the benefits of both to programmers. Static
  typing provides error detection and strong guarantees while dynamic
  typing enables rapid prototyping and flexible programming
  idioms. For programmers to take full advantage of a gradual type
  system, however, they must be able to trust their type annotations,
  and so runtime checks must be performed at boundaries between static
  and dynamic code to ensure that static types are respected. Higher
  order and mutable values cannot be completely checked at these
  boundaries, and so additional checks must be performed at their use
  sites. Traditionally, this has been achieved by installing proxies
  on these values to moderate the flow of data between static and
  dynamic code, but this can cause problems if the language supports
  comparison of object identity or foreign function interfaces.

  Reticulated Python is a gradually typed variant of Python 3
  implemented via a source-to-source translator. It implements a
  proxy-free design named \emph{transient casts}.  This paper presents
  a formal semantics for transient casts and shows that not only are
  they sound, but they work in an open-world setting in which the
  Reticulated translator has only been applied to some of the program;
  the rest is untranslated Python. We formalize this \emph{open world
    soundness} property and use Coq to prove that it holds for Anthill
  Python, a calculus that models Reticulated Python.
\end{abstract}
% \category{D.3.3}{Programming Languages}{Language Constructs and Features}
% \keywords
% gradual typing, translation, python, proxy

\section{Introduction}

Gradual typing~\cite{Siek:2006bh,Tobin-Hochstadt:2006fk} enables the
safe interaction of statically typed and dynamically typed code. In
recent years, gradual typing has been of interest not only to the
research community
\cite{Rastogi:2012fk,Takikawa:2012ly,Swamy:2014aa,siek:2015mono,Ren:2013vn,allende:2013gradualtalk,Ahmed:2011fk}
but also to industry, as numerous new languages have arrived on the
scene with features inspired by gradual typing, including Hack
\cite{hack} and TypeScript \cite{typescript}.

Gradually typed languages use the dynamic type $\dyn$ and a
\emph{consistency relation} on types to govern how statically typed
and untyped code interacts: types are consistent if they are equal up
to the presence of $\dyn$. The consistency relation replaces type
equality in the static type system. For example, at a function call
the type of an argument is required to be consistent with the
parameter type of the function, not equal to it. To prevent
uncaught runtime type errors, additional checks must be performed at
runtime. These checks are typically achieved by inserting casts as
part of a type-directed, source-to-source translation.

Designing semantics for these casts in the presence of higher-order
values and mutation is a significant research endeavor
(\citet{Herman:2006uq, Siek:2009rt, siek:2015mono}, etc.). The
traditional approach \cite{Findler:2002es,Siek:2006bh,Herman:2006uq}
installs runtime guards (i.e.\ wrappers or proxies) on values, which
moderate between differently typed regions of code. Here we refer to
this as the \guarded{} approach.  Alternatively, the \monotonic{}
approach of \citet{siek:2015mono} avoids proxies by using runtime type
information to lock down objects as they pass through casts. With the
monotonic approach, some programs that would execute without error
under \guarded{} instead produce a runtime cast error. The related
design of \citet{Swamy:2014aa} requires further restrictions. For
example, a function of the type $\dyn\to\dyn$ cannot be cast to 
$\mathsf{bool}\to\mathsf{bool}$.

\emph{Reticulated Python}\footnote{Named for the largestc species of snake in the
  Python family.} is a platform for experimenting with
different dialects of gradually typed
Python~\cite{vitousek:2014retic}. One design implemented in
Reticulated Python is the \transient{} approach, which uses pervasive
use-site checks in lieu of guarding objects with proxies. This
alternative avoids a serious issue for \guarded{}: proxies are not
pointer-identical to the underlying proxied object, a fact that
results in significant and challenging-to-debug problems in practice.
\citet{vitousek:2014retic} use Reticulated to examine the effects of
proxies on gradually typed Python programs and reproduced the problems
observed by \citet{cutsem:2013proxies}.

Although significant research has focused on leveraging gradual typing
to achieve performance improvements
\cite{Herman:2006uq,Wrigstad:2010fk,siek:2015mono}, efficiency is not
an aim of the \transient{} design. Instead, \transient{} is an
error-detecting technique, similar to the checked mode of Dart
\cite{dart} (but sound). It can be enabled during development and
debugging and disabled for release.

\paragraph*{Transient semantics modeled in Anthill Python}

In this paper we present a complete formal semantics for \emph{Anthill
  Python},\footnote{Named for the smallest species of snake in the
  Python family.} a gradually typed calculus which models Reticulated
Python and which uses the transient semantics. Anthill Python supports much of the complexity of the Python
object system, including bound methods, multiple inheritance, and
mutable class fields and methods.

Just as Reticulated programs are translated into Python and then
executed, the dynamic semantics of Anthill is defined by translation
into an untyped language, \emph{\upython}. We define the
dynamic semantics of \upython{} and a translation from Anthill to
\upython, inserting checks according to the \transient{}
semantics. While \citet{vitousek:2014retic} informally described
\transient{}, Anthill and \upython{} provide a mathematical
description of the design. 

\paragraph*{Open world gradual typing}

In addition, in this paper we show that the distinction between
\guarded{} and \transient{} affects more than object identity; it
affects whether or not the language supports sound interaction with
foreign code in an ``open world''.

Like many gradually typed languages, Reticulated Python translates
input programs to a target language with explicit casts, in this case
standard Python 3. When using \guarded{}, the translated code cannot
interact with untranslated Python code (in this paper referred to as
``plain Python'') without losing the usual type soundness guarantees
of gradual typing~\cite{Siek:2006bh}: plain Python programs do not
include the explicit casts necessary to maintain the invariants
expected by translated Reticulated programs.  The \transient{}
approach, on the other hand, does not make any assumptions about its
clients. Instead, typed code performs all required checks internally.

The \guarded{} approach poses additional problems in the context of
Python. CPython, the reference implementation of Python, is
implemented in C, and many built-in functions are defined as C
functions. Furthermore, Python supports C and C++ extension modules
--- modules callable from Python that are implemented in C. This code
does not respect the abstractions used by \guarded{} to ensure
soundness and can mutate Python data structures as raw memory,
possibly violating invariants specified by Reticulated's type
system. The \transient{} semantics can soundly coexist with such
foreign functions because it performs all required checks within typed
code at use sites.

In this paper we discuss a formal property called \emph{open world
  soundness} which states that, if a program is translated from a
gradually typed surface language into an untyped target, it can be
embedded inside arbitrary untyped code without ever being the source
of an uncaught type error. For gradually typed languages, open world
soundness is a stronger property than traditional type safety.
We prove that open world soundness holds for the transient design as
modeled by Anthill Python and its translation to \upython. Our proof
is machine-checked by the Coq proof assistant and can be found at
\hosturl.  It relies on an auxiliary static type system for \upython{}
that captures invariants about programs translated from Anthill. This
type system discriminates between translated code and plain \upython{}
code by tagging all reducible expressions with labels indicating the
origin of the expression. The type system places strong requirements
on the translated expressions, while untranslated code is
unrestricted. We prove that the translation from Anthill to \upython{}
is type-preserving.

%% Ultimately this allows sound and seamless interoperation between code
%% written in plain Python, and gradually typed code translated by
%% Reticulated.

%% extends type safety with the guarantee that
%% well-typed code won't go wrong even in the presence of arbitrary
%% interaction with untyped code.

%% The \transient{} approach supports , 

%% The open world soundness theorem 

Open world soundness also allows programmers to benefit from gradual
typing without using its features themselves. Using \transient, a
Reticulated library writer can use static types in their code,
including on API boundaries. If this library is then translated to
regular Python, any Python programmer who uses the library will
benefit from the improved error detection resulting from Reticulated's
static types without having to use --- or even know about ---
Reticulated itself.

% These issues are not exclusive to the context of Reticulated Python
% --- similar situations may occur in other languages, for example, a
% typed variant of JavaScript interacting with a JavaScript library
% downloaded from a remote server and then evaled.
\subsection*{Contributions}
\begin{itemize}
\item We define the property of \emph{open world gradual typing} and
  evaluate the \guarded{} and \transient{} designs with respect to it
  (Section \ref{sec:openworld}).
\item We provide a formal semantics for the \transient{} design in the
  form of Anthill Python, a calculus that models Reticulated Python
  (Section \ref{sec:anthill}). This semantics is defined by
  translation to an untyped calculus, \upython{}, that models Python.
\item We prove that Anthill exhibits \emph{open world soundness}, an
  extension of type soundness which states that the translation from
  Anthill to \upython{} produces code that doesn't go wrong even when
  interoperating with plain \upython{} code (Section
  \ref{sec:proofs}). The full proof in Coq is at \hosturl.
\end{itemize}
Section \ref{sec:review} discusses background and related work, and
Section \ref{sec:conclusions} concludes.

\section{Background and Related Work}\label{sec:review}

In this section we discuss related work and review the prior work on
Reticulated Python~\cite{vitousek:2014retic}.

\subsection{Gradual typing}
\label{sec:gt-relatedwork}

Researchers and language designers have been interested in mixing
static and dynamic typing in the same language for quite some time
\cite{cartwright:1976uddt, Abadi:1989ez,
  bracha:1993strongtalk}. Gradual typing, developed by independently
by \citet{Siek:2006bh} and \citet{Flanagan:2006mn}, was preceded
quasi-static typing of \citet{Thatte:1990yv}, the Java extensions of
\citet{Gray:2005ij}, Bigloo, \cite{Serrano:2002zo} and Cecil
\cite{Chambers:2004vt}. Gradual typing is distinguished by its use of
the consistency relation (originally defined for nominally typed
languages by \citet{Anderson:2002kd}) to govern where typed and
untyped code can flow into each other, and where runtime checks need
to be performed to ensure safety at runtime. See \citet{siek:2015gg}
for a detailed discussion of the core principles that gradual typing
aims to satisfy.

\citet{allende:2013strats} develop an alternate cast insertion scheme
for Gradualtalk \cite{allende:2013gradualtalk} that facilitates
interaction between typed libraries and untyped clients without
requiring client recompilation. % needing to recompile the clients.
This approach, called % which they refer to as the the ``execution
``execution semantics'' uses callee-installed proxies on function
arguments.  While similar to our semantics, this approach is still
vulnerable to the problems with proxies that we address with the
\transient{} approach (see Section~\ref{sec:proxies}).

Another relevant alternative design is the \emph{like-types}
approach~\cite{bloom:2009thorn,Wrigstad:2010fk}. This approach avoids
proxies by splitting static type annotations into \emph{concrete
  types} (whose inhabitants are never proxied, and which cannot flow
into dynamically typed code) and like types, which can freely interact
with dynamic code. This approach was used in designing a sound variant
of TypeScript called StrongScript, which obeys open-world soundness
\cite{richards:2015strongscript}. However, because of the
incompatibility of concrete and like types, it is not straightforward
to evolve a program in StrongScript from dynamic to static, as is
frequently desirable. As a result, the \emph{gradual guarantee} of
\citet{siek:2015gg} does not hold for these designs.

Typed Racket~\cite{samth:2008ts,Takikawa:2012ly} includes first-class
classes and supports open-world interaction between Typed Racked
modules and untyped Racket code. In that work, the authors are aided
by Racket's module semantics, which allows functions exported from
Typed Racket into untyped code to be wrapped with a contract monitor
that ensures that interaction between typed and untyped code is sound
even though untyped clients are unaware of the static type
system. These features allow Typed Racket to have open world soundness
with a \guarded{} approach. Python does not have the same capabilities
for controlling module exports and faces the problems with proxies and
foreign functions explained above.

In recent years, gradual typing, or ideas related to it, has become
popular among industrial language designers, with C\#
\cite{Bierman:2010fk} adding a dynamic type and Typescript
\cite{typescript}, Dart \cite{dart}, and Hack \cite{hack} offering
static typechecking of optional type annotations. Academic language
designers have also retrofitted gradual typing to existing languages,
such as Racket \cite{samth:2008ts}, Smalltalk
\cite{allende:2013gradualtalk}, Ruby \cite{Ren:2013vn}, and Python
\cite{vitousek:2014retic}.

\subsection{Reticulated Python and the guarded cast semantics}

Reticulated Python implements several designs for gradual typing,
including \guarded, the traditional design for gradual
typing~\cite{Siek:2006bh,Herman:2006uq}. In this design, casts are
inserted at the locations of implicit conversions.  For example,
consider the following statically typed \li{filter} function, which
expects a parameter of type \li{Callable([int], bool)} (the syntax for
function type $\mathsf{int}\to\mathsf{bool}$).

\begin{minipage}{1.0\linewidth}
\begin{lstlisting}
def filter(fn:Callable([int],bool), lst:List(int))->List(int):
  nlst = []
  for elt in lst:
    if fn(elt):
      nlst.append(elt)
  return nlst
filter(lambda x: x % 2 == 0, [1, 2, 3, 4]) (*\label{ex:ci:call}*)
\end{lstlisting}  
\end{minipage}

The function created on line \ref{ex:ci:call} has no type
annotations, so it has type \li{Callable([Dyn], Dyn)}, where \li{Dyn}
is the dynamic type. It is casted from \li{Callable([Dyn], Dyn)} to
\li{Callable([int], bool)} at runtime, because of a cast inserted at
line \ref{ex:ci:call} by the translation from Reticulated Python to
Python.

Casts on first-order values are checked immediately but casts on
functions and objects are not. With \guarded{}, function casts
install a wrapper on the function which, when called, casts the input,
calls the underlying function, and casts the
result~\cite{Findler:2002eu}. Casts on objects return a proxy --- a
new object that casts fields and methods during reads and writes.

To see how object proxies are used to preserve type soundness, consider the
example in Figure~\ref{fig:guard-errors}, where the object
\li{Foo()}\note[Mike]{Space allowing, show translated code} is
expected to have type \li{\{'bar':int\}} by the typed \li{f} function. The \li{f} function passes
the object to the untyped \li{g} function which mutates the \li{bar}
field to contain a string. Upon return from \li{g}, \li{f} accesses
the \li{bar} field, expecting an \li{int}.
The code highlighted in \colorbox{\cicolor}{yellow} indicates the
locations where proxies are installed. The object bound to the
parameter \li{x} in \li{g} is therefore not the object passed to
\li{g} by \li{f} at line \ref{ex:rg:call}, but is instead a proxy to
that object. When \li{g} attempts to write a string to that proxy at
line \ref{ex:rg:upd8}, the proxy casts it to \li{int}, the type that
\li{y.bar} is expected to have in \li{f}, which triggers a cast error.

\begin{wrapfigure}{r}{.4\textwidth}
\begin{lstlisting}
class Foo:
  bar = 42
def g(x):
  x.bar = 'hello world'(*\label{ex:rg:upd8}*)
def f(y:{'bar':int}):
  g((*\makebox[0pt][l]{\color{\cicolor}\rule[-0.5ex]{0.5em}{2.1ex}}*)y) (*\label{ex:rg:call}*)
  return y.bar (*\label{ex:rg:deref}*)
f((*\makebox[0pt][l]{\color{\cicolor}\rule[-0.5ex]{2.5em}{2.2ex}}*)Foo()) (*\label{ex:rg:inst}*)
\end{lstlisting}
  \caption{Soundness via proxies}
  \label{fig:guard-errors}
\end{wrapfigure}

\subsection{Guarded faces practical problems}
\label{sec:proxies}

The \guarded{} approach is well studied, with many optimizations for
space and time efficiency and additional features such as blame
tracking~\cite{Herman:2006uq, Siek:2012uq, Ahmed:2011fk,
  Wadler:2009qv}, but a significant problem with \guarded{} is that a
proxy is not identical to its proxied object. This issue leads to
unexpected behavior when pointer equality checks (using Python's
\li{is} operator) are performed. \citet{cutsem:2013proxies} also
encountered and discussed this issue.  \citet{vitousek:2014retic}
found these issues to be a significant problem in Reticulated Python
in \guarded{}, preventing many programs from running as expected.

A related issue is that Python programs can inspect the class of a
value using the \li{type} function, and the class of a proxy is a
proxy class (rather than the class of the proxied object). Programs
that use reflection over the types of values can execute incorrectly
when the \li{type} function returns a proxy class rather than the
expected class.

\subsection{The transient cast semantics}

To solve this issue, \citet{vitousek:2014retic} implemented an
alternative, called the \transient{} approach, that does not use
proxies. (\citet{vitousek:2014retic} report on the implementation of
\transient{} but they do not define a formal semantics.)  The
\transient{} design uses \emph{checks} --- lightweight queries about
the current state of a value --- rather than proxy-building
casts. Checks are inserted into programs at function calls and object
reads. This approach transfers the responsibility for checking object
reads and function return values from the object or function itself to
the context that is performing the read or call. From an
implementation perspective, this moves the checking from the runtime
system (e.g.\ casts performed by object proxies) into the translated
programs' code.

Figure \ref{fig:ci} illustrates this difference. Figure \ref{fig:ci}a
shows a program in which a function \li{f} takes and uses a typed
object parameter. The argument to the function comes from untyped
code, and thus any sound semantics for gradually typing will require
some runtime checks when running this program. Text in
\colorbox{\pcolor}{yellow} indicates locations where checks must
occur because program values are entering into a context where they
are expected to have a certain type --- either by being passed in as
arguments, dereferenced from mutable state, or returned from a
function call.

Figure \ref{fig:ci}b shows the same program after explicit runtime
checks have been inserted. Checks are performed pervasively, including
in statically typed code. The write at line \ref{ex:ti:write} does not
need a check because \li{obj} is not proxied and any reads of \li{x}
will ensure that it has the right type for its context. The
highlighted parameters on line \ref{ex:ti:def1} indicate that the
function checks that \li{z} and \li{obj} have their expected types at
its entry point, translated as the explicit checks at the beginning of
the function body in Figure \ref{fig:ci}b. The check on line
\ref{ex:ti:get} verifies that \li{obj.meth} has a function type, and
line \ref{ex:ti:call} checks that the result of the function call is
an \li{int}. This transfer of responsibility from clients (and the
proxies they install) to typed code allows this design to succeed
without needing proxies.

To see how \transient{} preserves soundness in the presence of
mutation, consider the example in Figure \ref{fig:errors}, which
contains the same program as Figure \ref{fig:guard-errors}. Again,
sites where transient checks are made are highlighted in
\colorbox{\cicolor}{yellow}. While \guarded{} would install a proxy on
\li{y} at line \ref{ex:rt:call} to ensure that \li{y.bar} remains an
\li{int}, the \transient{} semantics passes \li{y} into \li{g}
directly. The strong update at line \ref{ex:rt:upd8} proceeds, but the
fact that a string has been written to \li{y.bar} will not be
observable within \li{f}. This is because at the dereference of
\li{y.bar} at line \ref{ex:rt:deref}, a transient check will attempt
to verify that \li{y.bar} is an \li{int} (the type expected in that
context). When this check fails, a runtime check error is reported.

\makeatletter
\let\origthelstnumber\thelstnumber
\newcommand*\Suppressnumber{%
  \lst@AddToHook{OnNewLine}{%
    \let\thelstnumber\relax%
     \advance\c@lstnumber-\@ne\relax%
    }%
}

\newcommand*\Reactivatenumber{%
  \lst@AddToHook{OnNewLine}{%
   \let\thelstnumber\origthelstnumber%
   \advance\c@lstnumber\@ne\relax
 }%
}
\makeatother

\begin{figure*}[tbp]
  \begin{multicols}{2}
    (a) Pre-insertion:
    \begin{lstlisting}
def f((*\makebox[0pt][l]{\color{\cicolor}\rule[-0.5ex]{10em}{2.3ex}}*)z:int, obj:{'x':int,(*\label{ex:ti:def1}\Suppressnumber*)
            (*\makebox[0pt][l]{\color{\cicolor}\rule[-0.5ex]{13.5em}{2.3ex}}*)'meth':Callable([int],int)})
            ->int:(*\label{ex:ti:def2}\Reactivatenumber*)
  obj.x = 42(*\label{ex:ti:write}*)
  method = (*\makebox[0pt][l]{\color{\cicolor}\rule[-0.5ex]{4.1em}{2ex}}*)obj.meth(*\label{ex:ti:get}*)
  result = (*\makebox[0pt][l]{\color{\cicolor}\rule[-0.5ex]{4.5em}{2ex}}*)method(z)(*\label{ex:ti:call}*)
  return result

def g(y):
  (*\makebox[0pt][l]{\color{\cicolor}\rule[-0.5ex]{4em}{2ex}}*)f(42, y)
    \end{lstlisting}
\columnbreak
    (b) Post-insertion:
%% def f(z:int, obj:{'x':int,
%%             'meth':Callable([int],int)})->int:(*\Suppressnumber*)
    \begin{lstlisting}
def f(z, obj):(*\Suppressnumber*)
  check(z, int)
  check(obj,{'x':int,
           'meth':Callable([int],int)})(*\Reactivatenumber*)
  obj.x = 42
  method = check(obj.meth, Callable([int],int))
  result = check(method(x), int)
  return result

def g(y);
  return check(f(42, y), int)
    \end{lstlisting}
  \end{multicols}
  \caption{Check insertion by the \transient{} semantics}
  \label{fig:ci}
\end{figure*}

\begin{wrapfigure}{l}{.4\textwidth}
%\hrule
\begin{lstlisting}
class Foo:
  bar = 42
def g(x):
  x.bar = 'hello world'(*\label{ex:rt:upd8}*)
def f((*\makebox[0pt][l]{\color{\cicolor}\rule[-0.5ex]{6.5em}{2.2ex}}*)y:{'bar':int}):
  g(y) (*\label{ex:rt:call}*)
  return (*\makebox[0pt][l]{\color{\cicolor}\rule[-0.5ex]{2.5em}{2.1ex}}*)y.bar (*\label{ex:rt:deref}*)
f(Foo()) 
\end{lstlisting}
  \caption{Soundness via transient}
  \label{fig:errors}
\end{wrapfigure}

The lack of proxies and the reliance on use-site checks in the
\transient{} approach means that many of the problems with \guarded{}
are solved immediately. Object identity behaves normally because
objects and their types are always the same as in plain Python, and
every check that does not fail simply returns the checked value. On
the other hand, checks are pervasively installed even within typed
code, and in places where they would not be needed under
\guarded{}. Thus \transient{} is primarily useful as a debugging tool,
rather than something to be deployed in production code.

Other approaches also tackle the problem of object identity:
\citet{Keil:2015aa} present a solution based on the idea of making
proxies transparent with respect to identity and type
tests. TypeScript \cite{typescript}, Dart \cite{dart}, and other
languages that compile to JavaScript without any runtime checks are
trivially free of proxies but their type annotations are not enforced
at runtime. Dart also offers a \emph{checked mode}, wherein function
arguments are checked at runtime against optional type annotations,
similar to the \transient{} approach, but these checks do not cover
all cases and uncaught runtime type errors are still possible.
However, \transient{} solves an additional diffuculty that \guarded{}
faces: interoperability with untyped or foreign code in an open world,
which we discuss in the next section.

\section{Open World Gradual Typing}\label{sec:openworld}

In this section, we discuss the property of open world gradual
typing, show that \guarded{} is an obstacle to it, and discuss the
solution provided by \transient. While this discussion is specific to
Python, many issues are broadly applicable to other dynamically typed
languages such as JavaScript and Clojure.

\subsection{Guarded is unsuitable for an open world}

While \guarded{} is sound when all parts of a program have been seen
by the typechecker and translator, it is not always desirable or even
possible for that to be the case. We examine two situations where
interactions between Reticulated programs and plain Python or foreign
functions are troublesome.

\paragraph*{Plain Python clients cannot soundly use Reticulated programs with Guarded}

Reticulated Python typechecks programs and translates them into Python
3, and once they have undergone this translation, they can be executed
by a Python interpreter. As such programs are valid Python programs,
we would like to use them in combination with plain Python programs
from different sources. Unfortunately, when using Reticulated with
\guarded{}, this is not type safe and can cause programs to behave
unexpectedly.

It is important here to distinguish plain Python programs from
\emph{untyped Reticulated code}. The latter is Reticulated Python code
that does not contain type annotations, but that is translated into
Python by Reticulated. This translation is \emph{not} an identity
transformation, even on untyped Reticulated code: if the untyped code
makes calls into typed code, it must ensure that the values it passes
into typed code are cast to their expected type.

\begin{figure}[tbp]
\noindent\begin{minipage}{\textwidth}
\begin{minipage}[c][][c]{\dimexpr0.5\textwidth-0.5\Colsep\relax}
  Module \textbf{utils} (Reticulated Python):
  \begin{lstlisting}
def is_odd(x:int)->int:
  return x % 2
def intappend(lst:List(int), x)->List(int)
  list.append(lst, x)
  return lst
  \end{lstlisting}
\end{minipage}\hfill
\begin{minipage}[c][][c]{\dimexpr0.5\textwidth-0.5\Colsep\relax}
Module \textbf{fastexp} (Reticulated Python):
%% import utils
%% def fastexp(x, n):
%%   if n == 0:
%%     return 1
%%   elif n == 1:
%%     return x
%%   elif utils.is_odd(n):
%%     return x * fastexp(x * x, (n-1) // 2)
%%   else: 
%%     return fastexp(x * x, n // 2)
\begin{lstlisting}
import utils
def fastexp(x, n):
  ... utils.is_odd(n) ...
\end{lstlisting}
Module \textbf{interactive} (Plain Python):
\begin{lstlisting}
import utils
while True:
  print('Enter a number')
  inp = input()
  print(utils.is_odd(inp))
\end{lstlisting}
\end{minipage}%
\end{minipage}
\begin{center}
{\centering\hrule
\begin{tikzpicture}[scale=1, every node/.style={transform shape}]
%   \fill [fill=blue] (2,-0.8) rectangle (4.4, -2.7);
   \node       (Ty1p)      {$\lfloor\textbf{utils}\rfloor$};
   \node[right=of Ty1p]       (Ty2p)      {$\lfloor\textbf{fastexp}\rfloor$}
     edge[pil, <->] (Ty1p);
   \node[above=of Ty1p]        (Ty1)       {\bf utils}
     edge[pil, dashed]  node[left] {$\leadsto$} (Ty1p);
   \node[above=of Ty2p]       (Ty2)      {\bf fastexp}
     edge[pil, <->] (Ty1)
     edge[pil, dashed] node[left] {$\leadsto$} (Ty2p);

   \node[left=of Ty1p]        (Un)    {\bf interactive}
     edge[pil, ->] (Ty1p);
   \node[below=of Un]         (Nat)   {\li{list.append}}
     edge[pil, <->] (Ty1p); 
   
   \node[above=1mm of Ty1] (V1) {Typed};
   \node[above=1mm of Ty2] (V2) {Untyped};
   \node[left=of V1] (V3) {Foreign};

   \node[right=2mm of Ty2p] (T2) {Python};
   \node[above=of T2] (T1) {Retic.};
   \node[below=of T2] (T3) {Native};
%   \draw[->,line width=1.5pt]   (Usr1) to (Mon1);
\end{tikzpicture}
}
  
\end{center}
  \caption{Reticulated and plain Python interaction}
  \label{fig:interaction}
\end{figure}
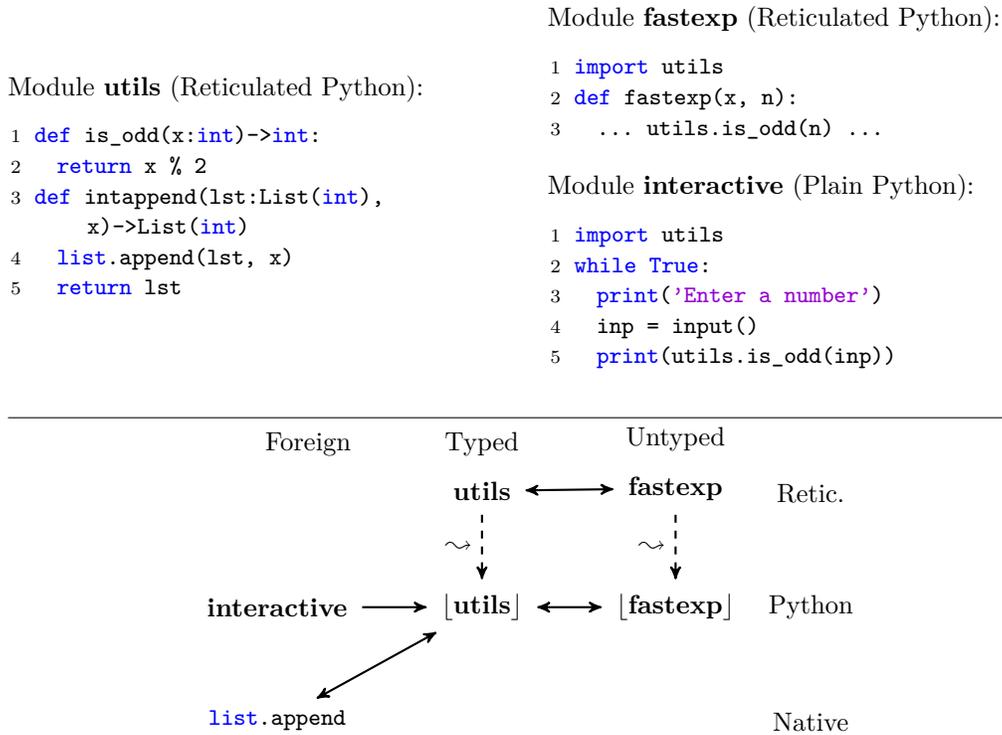

Figure \ref{fig:interaction} shows interaction between translated
Reticulated and plain Python, which is problematic under the
\guarded{} semantics. Two modules, \textbf{utils} and
\textbf{fastexp}, can communicate freely with each other even though
\textbf{fastexp} has no type annotations, because Reticulated
translates them (the dotted~$\leadsto$~arrows at the bottom of Figure
\ref{fig:interaction}) both into Python programs
$\lfloor\textbf{utils}\rfloor$ and $\lfloor\textbf{fastexp}\rfloor$
with explicit casts. Because $\lfloor\textbf{utils}\rfloor$ is a
standard Python module, the plain Python module \textbf{interactive}
may attempt to use it. The \textbf{interactive} module passes in a
string (the result of \li{input}) where the type system expected an
\textsf{int}, but this error is not detected because, in the
\guarded{} approach, it was the responsibility of \textbf{interactive}
to provide the correct type by checking its arguments at the call site
(or, in the case of higher order values, installing a proxy around
them). The result is a confusing Python error:

\begin{minipage}{1.0\linewidth}
\begin{lstlisting}[language={}]
  File "utils.py", line 2, in is_odd
    return x % 2
TypeError: not all arguments converted during string formatting
\end{lstlisting}
\end{minipage}

\noindent
This is the same error that would occur in plain Python without static
types, but now it may be more difficult to debug: the programmer may
not realize that they cannot trust the type annotations.

This issue prevents programmers from writing Reticulated programs with
types on API boundaries and distributing them as servers for unknown,
plain Python clients. Because typed API boundaries are an important
use case for gradual typing, this is a significant problem that
reaches well beyond Python. This can be partially solved by providing
untyped API functions for plain Python clients and casting the
client's calls to static types internally (an option discussed in the
context of Gradualtalk by \citet{allende:2013strats}). However, this
alone cannot guarantee safety in Reticulated's context, because Python
modules cannot selectively export information to clients --- clients
can still access the typed functions.

\paragraph*{Foreign functions cannot operate on proxies}
\label{sec:ffprox}
Another issue arises because many of the built-in functions used by
Python programs are not defined in Python itself, but rather in
compiled binaries. For example, in CPython the \li{append} method for
lists is defined in C code which is not accessible for introspection
by Reticulated and which, more importantly, does not use the same
machinery for attribute access as regular Python code.  

This is illustrated by the \li{intappend} function in \textbf{utils}
in Figure \ref{fig:interaction}. Suppose \li{intappend} is called from
Reticulated code as follows: \vspace{1ex}
 
\begin{minipage}{1.0\linewidth}
\begin{lstlisting}
import utils
utils.intappend([1,2,3],4) (*\label{ex:append:call}*)
\end{lstlisting}
\end{minipage}

\noindent
Under the \guarded{} semantics, at the call to \li{utils.isodd} at
line \ref{ex:append:call}, the list \li{[1,2,3]} is wrapped in a
proxy, which should prevent code from writing non-\li{int} values to
it. However, when it is passed to \li{list.append} as the
receiver,\footnote{For those unfamiliar with Python, \li{list.append}
  is an unbound method of the \li{list} class, and its first argument
  is set as the receiver of the call. By calling \li{list.append(x,
    ...)} instead of \li{x.append(...)}, we avoid going through the
  proxy on \li{x}.} \li{list.append} is not operating on the list
\li{[1,2,3]}, but instead on a proxy to it.

Like any Python value, the proxy has to be an instance of some
class. If the proxy's class was unrelated to \li{list}, then a list's
proxy would not appear to be a list instance to the commonly-used
\li{isinstance} function, causing many unexpected errors. Thus in
Reticulated's implementation of \guarded{}, this proxy will be a value
whose class is a dynamically generated subclass of the \li{list}
class. However, this choice means that the proxy itself contains its
own list in the heap, distinct from the list it is proxying. As a
proxy, its methods are replaced with reflective calls to the
underlying, proxied list, so this internal list is inaccessible to
Python code. The native code that implements \li{list.append},
however, directly operates on the memory in the heap containing the
list data of the proxy, bypassing the proxy's methods. As a result,
the proxy cannot intercept and forward changes to the underlying list,
and thus \li{list.append} mutates the (empty and otherwise
inaccessible) list contained within the proxy itself and leaves the
proxied list untouched. The overall result is that, quite
surprisingly, the result of the call at line \ref{ex:append:call} is
\li{[1,2,3]} --- the list that was passed in, unmodified.

As a result, Reticulated programs under \guarded{} can
behave unpredictably and in ways contrary to their expected
semantics. This problem is difficult to diagnose and its scope is
difficult to determine.

\subsection{The transient cast semantics is sound in an open world}
\label{sec:inf:trans}

Though \citet{vitousek:2014retic} discuss \transient{} only as a
technique to solve issues with type identity, it also enables open
world soundness. \texttt{Transient} transfers the responsibility for
checking object reads and function return values from the object or
function itself to the context that is performing the read or
call. Therefore, if plain Python code passes ill-typed values into
translated code, checks within the typed module will detect a type
mismatch and report an error rather than causing a confusing error or
not reporting the error at all. In Figure \ref{fig:interaction},
when \li{[1,2,3]} is an argument of \li{list.append}, it actually is
the list \li{[1,2,3]} being passed in, and native code can operate on
it exactly as it normally would. Even if native code mutates the list
in an ill-typed way, it will be detected as soon as a dereference
occurrs within typed Reticulated code.

Furthermore, because \transient{} avoids proxies, C code sees only the
Python values that were expected. When the \textbf{interactive} module
calls \textbf{utils}, instead of a confusing error referring to string
formatting, the result will be a cast error saying that a string has
been passed in instead of an int.

\section{Formalizing Transient with Anthill Python}
\label{sec:anthill}

The above section illustrates some of the complexities of integrating
typed and untyped code in Reticulated Python, and the ways that
Reticulated code can be used in an open world. This section formalizes
these properties, and the \transient{} design, with two calculi: a
gradually typed surface language named Anthill Python that models
Reticulated and an a calculus named \upython{} that models plain
Python.

\subsection{Design considerations of the type system} 

In this section we discuss some of the distinct features of Python
that Anthill Python models and how these features inform the design of
Anthill's static type system.

\paragraph*{Python classes}

Like many dynamic languages, Python's classes are first-class values
which can be passed between program components like any other kind of
data. Therefore a static type system for Python must give types to
classes just as it does to objects and functions. In Python, classes
may also directly contain attributes which can be accessed and
mutated, and they can be extended with new fields.  Such mutations to
a class value can affect all its instances.  For example, one can add
methods to a class after its creation and then call those methods on
any instance of the class.
Classes are also themselves callable --- calling a class is the Python
syntax for object construction.

These observations indicate that classes share elimination forms with
both objects and functions. To typecheck call sites and attribute
reads or writes without undue restrictions, the class types of
Reticulated (and Anthill) are subtypes of both object and  function types.

\paragraph*{Object construction}
\label{sec:statics:dynamic} 

Python objects come into existence with their object-local fields
uninitialized. The user-defined code in the \li{__init__} method must
instanciate them by mutating the receiver of the method (represented
as the first parameter of the method, called \li{self} by convention).
The \li{__init__} method is not restricted to only initializing
fields, though --- it can contain arbitrary Python code, including
code that allows the partially initialized receiver to escape into
other functions.  Consider the following code in which an object is
specified to have \li{x} and \li{y} fields of type \li{int} by the
\li{@fields} decorator.  The result is a failed check at line
\ref{ex:init:read}.

\begin{minipage}{1.0\linewidth}
  \begin{lstlisting}
@fields({'x':int, 'y':int})
class Point2D:
  def __init__(self:Point2D, x:int, y:int):(*\label{ex:init:decl}*)
    print(magnitude(self))
    self.x = x
    self.y = y
def magnitude(pt:Point2D)->float:
  return math.sqrt(pt.x ** 2 + pt.y ** 2) (*\label{ex:init:read}*)
  \end{lstlisting}
\end{minipage}

The problem here is that the annotation on line \ref{ex:init:decl}
that \li{self} is a \li{Point2D} is not true: at that point, \li{self}
does not obey the type that a \li{Point2D} is expected to have,
lacking the fields \li{x} and \li{y}. This problem is well known in
other object-oriented languages such as Java, and several solutions
have been proposed~\cite{qi:2009masked, fahndrich:2007delayed}. These
approaches are substantially more complex than the rest of the
Reticulated type system, and by virtue of being gradually typed,
Reticulated can rely on dynamic type checks to preserve
soundness. Reticulated therefore requires that the receiver parameter
of a constructor, i.e.\ \li{self}, is untyped. This behavior matches
Python, but it means that no program with objects can ever be truly
fully statically typed --- every object goes through a dynamic
initialization phase, after which it can be injected to a static
type by inserting a check at object instantiation sites.

\newcommand{\tsym}{A}
\begin{figure}
  \centering
  \[
  \begin{array}{lcl}
    t & ::= & x \mid n \mid t(\overline{t}) \mid t.\ell \mid t.\ell=t \mid \lett{x}{t}{t}\mid 
    \classt{X}{q}{\Delta}{\Delta}{\overline{t}}{\overline{\ell{=^M} m}}{\overline{\ell{=^F}t}}{c} \mid \\&& \funt{\overline{x{:}\tsym}}{\tsym}{t}\\
    m & ::= & \method{\overline{x{:}\tsym}}{\tsym}{t}\\
    c & ::= & \constructt{\overline{x{:}\tsym}}{t}\\
    q & ::= & \lozenge \mid \blacklozenge \\
    \tsym & ::= & \dyn \mid \mathsf{int} \mid \overline{\tsym} \to \tsym \mid \classty{X}{q}{\Delta}{\Delta}{\overline{\tsym}} 
    \mid \objty{X}{q}{\Delta}\\
    \Delta & ::= & \overline{\ell:\tsym}
  \end{array}
  \]
  \caption{Syntax for Anthill Python. (Overlines indicate sequencing.) }
  \label{fig:anthill-syntax}
\end{figure}

\subsection{Anthill Python}

Anthill Python is a gradually typed calculus which models Reticulated
Python.  The syntax of Anthill Python is defined in Figure
\ref{fig:anthill-syntax}.

\subsubsection{Types} 

Anthill Python's types include the dynamic type $\dyn$, integers,
$n$-ary function types, and object and class types. 

Object types $\objty{X}{q}{\Delta}$ contain two parts: an
\emph{openness} descriptor $q$ and an \emph{attribute type}
$\Delta$. Attribute types, written are partial maps from attribute
names to types, and represent the structural type of the object. An
openness descriptor can be either $\lozenge$ for ``open'' or
$\blacklozenge$ for ``closed''. Open objects support implicit width
downcasts, while getting or setting fields not present in $\Delta$
causes a static type error in a closed class; see Section
\ref{sec:transl:rw}.

Class types contain an openness descriptor and two attribute
types. The first attribute type contains the public interface for the
class itself while the second represents the instance fields present
in every instance of the class, but not in the class itself. The
instance fields are used in the type system to derive object types
from class types and the presence of instance fields is checked at
runtime after object construction. Class types also record the
parameter types of their constructors, so that object instantiations
can be type checked.

\subsubsection{Terms}

Terms in Anthill Python include variables, numbers, applications,
attribute access and update, and $n$-ary functions with parameter and
return type annotations. Methods and fields are read from with
$t.\ell$ and written to with $t.\ell=t$, where $\ell$ is the name of an
attribute; there is no syntactic differentiation between accessing
fields and methods.

Anonymous class definitions are also terms, and contain a number of
components. The openness descriptor $q$ and the attribute types
$\Delta_1$ and $\Delta_2$ specify the overall static type of the class.
The class' constructor is given by $\mathtt{init}=c$. In Python (and
Reticulated), object constructors are simply methods named
\li{__init__}, but to simplify the semantics, Anthill constructors $c$
are a special syntactic form. Constructors lack a return type, because
they are only invoked during object construction, and they perform
initialization by mutating a fresh, empty object passed to the
constructor.

The superclasses of a class definition are specified by the
$\overline{t}$ immediately after the \texttt{class} keyword. Fields
and their initializing terms are defined by
$\overline{\ell=^Ft}$. Methods $m$ (specified in class definitions by
$\overline{\ell=^Mm}$) are similar to functions except
that they always have an explicit receiver $x_s$ as their first
parameter. The $F$ and $M$ superscripts help to syntactically
distinguish between fields and methods.  No type annotation is given
for the receiver of a method; it is treated as having the type of an
instance of the enclosing class, as is typical for object-oriented
calculi~\cite{Abadi:1996fk}. Fields and methods are members of the
class, though they are accessible from instances, and constructors can
define fields specific to an instance.

%% Constructors themselves do not
%% return the constructed object, but instead are passed a reference to a
%% fresh, empty object and then mutate it. Anything returned from a
%% constructor is discarded.

\begin{wrapfigure}{R}{.6\textwidth}
  \centering
  \[
  \begin{array}{lcl}
    e & ::= & x \mid n \mid e(\overline{e}) \mid e.\ell \mid  e.\ell=e \mid \lett{x}{e}{e} \mid 
     \\&& \classe{X}{\overline{e}}{\overline{\ell=e}}{e}{} \mid 
     \fune{\overline{x}}{e} \mid \tcheck{e}{S}~\mid a\\
    S & ::= & \bdyn \mid \mathsf{int} \mid \bto{n} \mid \bclassty{\delta}{C} 
    \mid  \bobjty{\delta}\\
    C & ::= & n \mid \mathsf{Any}\\
    \delta & ::= & \overline{\ell}
  \end{array}
  \]
  \caption{Syntax for \upython}
  \label{fig:upython-syntax}
\end{wrapfigure}

\subsection{\upython}

The syntax of the target language of our translation, \upython, is
given in Figure \ref{fig:upython-syntax}. \upython{} differs
significantly from the ``classic'' cast calculi used by gradually
typed calculi, such as
$\lambda_\to^{\langle\tau\rangle}$~\cite{Siek:2006bh}. Typically, such
targets are statically typed languages that use explicit casts to
inject to and project from the dynamic type. In contrast, \upython{}
is dynamically typed, like Python itself, and has an expression that
performs \transient{} checks.
\newpage
\subsubsection{Expressions} 

The expressions of \upython{} mostly correspond to the terms of
Anthill. Class declarations do not syntactically distinguish between
functions, methods, and fields, so all class members are defined
together, and constructors still appear as a separate entry in the
class declaration but are now simply expressions.
The only new expressions introduced in \upython{} are heap addresses
and the check expression, $\tcheck{e}{S}$, which checks that the
runtime type of $e$ after evaluation corresponds to a \emph{type tag}
$S$.

%% \annote[Jeremy]{This aspect of Anthill departs slightly from the
%%   actual implementation of Reticulated. Reticulated programs are
%%   translated to include type checks, but rather than being first-class
%%   expressions in Python, these checks are calls to a straightforward
%%   Python function that performs the checking, an excerpt of which is
%%   shown in Figure}{Who cares}\note[Mike]{Leaving this in for now since
%%   it might be relevant with expanded guarded vs transient
%%   implementation stuff} \ref{fig:has-type}. To simplify \upython,
%% however, our formalization reifies these checks into an expression in
%% \upython.

%% \begin{figure}
%%   \begin{lstlisting}
%% def has_type(val, ty):
%%   if isinstance(ty, Dyn):
%%     return True
%%   elif isinstance(ty, Int):
%%     return isinstance(val, int)
%%   elif isinstance(ty, Object):
%%     return all(hasattr(val, member) and has_type(getattr(val, member), ty.member_type(member)) for member in ty.members)
%%   elif isinstance(ty, Function):
%%     return callable(val)
%%   elif ...
%%   \end{lstlisting}
%%   \caption{The \li{has_type} function, used by Reticulated's inserted checks, which corresponds to \upython's $\tcheck{e}{S}$ (from \citet{vitousek:2014retic})}
%%   \label{fig:has-type}
%% \end{figure}

\subsubsection{Type tags} 

The runtime type checks of \upython{} check against type tags $S$.
These tags represent information that can be checked \emph{shallowly}
and \emph{immediately} about \upython{} values. The tag \bdyn~is
satisfied by all objects, \textsf{int} is satisfied by integers, and
$\bto{n}$ is satisfied by functions of $n$ parameters, but makes no
claims about the types of the parameters or return type of the
function. Object and class tags $\delta$ contain the set of attributes
for their values, but not the types of their contents. Class tags also
optionally specify the arity of their constructors.

\subsection{Translation from Anthill to \upython}

Following the typical gradual typing approach, Anthill's runtime
semantics are defined in terms of a type-directed translation into
\upython. This translation rejects programs that are statically
ill-typed and inserts runtime checks according to the \transient{}
design.  As discussed in Section~\ref{sec:inf:trans}, \transient{}
inserts checks on all use sites of mutable and higher-order terms,
such as function applications or attribute accesses. The translation is
designed to ensure that, if a Anthill term $t$ has type $A$, then its
\upython{} translation $e$ can be checked against the type tag
$\flr{A}$ (defined in Figure \ref{fig:relations}) without resulting in
an error.

In this section, we discuss the details of this translation with
reference to excerpts of the translation relation. This translation is
an algorithmic, syntax-directed transformation. The translation is
specified in full in Appendix
\ref{apx:full-semantics}, Figure \ref{fig:apx:full-translation}.

\subsubsection{Attribute reads and writes}
\label{sec:transl:rw}

\begin{figure}
\boxed{\cig\vdash t\leadsto e:\tsym}
  \begin{mathpar}
\ninference{IGet}{\cig\vdash t\leadsto e:\tsym_1 \\ \Delta=\textit{mems}(\tsym_1) \\ \Delta(\ell)=\tsym_2}
          {\cig\vdash t.\ell \leadsto (\tcheck{e.\ell}{\flr{\tsym_2}}):\tsym_2}
\and
\ninference{IGet-Check}{\cig\vdash t\leadsto e: \tsym_1 \\ \Delta=\textit{mems}(\tsym_1)\\\\ \ell\not\in\textit{dom}(\Delta) \\ \textit{queryable}(\tsym_1)=\lozenge}
          {\cig\vdash t.\ell \leadsto (\tcheck{e}{\bobjty{\ell}}).\ell:\dyn}
\and
\ninference{ISet}{\cig\vdash t_1\leadsto e_1:\tsym_1 \\ \cig\vdash t_2 \leadsto e_2:\tsym_2' \\\\ \Delta=\textit{mems}(\tsym_1) \\ \Delta(\ell)=\tsym_2 \\ \tsym_2' \lesssim \tsym_2}
          {\cig\vdash t_1.\ell=t_2\leadsto (e_1.\ell=(\tcheck{e_2}{\flr{\tsym_2}})):\mathsf{int}}
\and
\ninference{ISet-Check}{\cig\vdash t_1 \leadsto e_1:\tsym_1 \\\\ \cig\vdash t_2\leadsto e_2:\tsym_2 \\ \Delta=\textit{mems}(\tsym_1) \\\\ \ell\not\in\textit{dom}(\Delta) \\ \textit{queryable}(\tsym_1)=\lozenge}
          {\cig\vdash t_1.\ell=t_2\leadsto (\tcheck{e_1}{\bobjty{\emptyset}}).\ell=e_2:\mathsf{int}}
  \end{mathpar}

\noindent\begin{minipage}{\textwidth}
\begin{minipage}[c][][c]{\dimexpr0.5\textwidth-0.5\Colsep\relax}
  \boxed{\textit{mems}(\tsym)=\Delta}\vspace{-3ex}
\[
  \begin{array}{rcl}
    \textit{mems}(\dyn) & = & \emptyset\\
    \textit{mems}(\objty{X}{q}{\Delta}) & = & \Delta\\
    \textit{mems}(\classty{X}{q}{\Delta_1}{\Delta_2}{\overline{\tsym}}) & = & \Delta_1
  \end{array}
\]
\end{minipage}\hfill
\begin{minipage}[c][][c]{\dimexpr0.5\textwidth-0.5\Colsep\relax}
  \boxed{\textit{queryable}(\tsym)=q}\vspace{-3ex}
\[
  \begin{array}{rcl}
    \textit{queryable}(\dyn) & = & \lozenge\\
    \textit{queryable}(\objty{X}{q}{\Delta}) & = & q\\
    \textit{queryable}(\classty{X}{q}{\Delta_1}{\Delta_2}{\overline{\tsym}}) & = & q
  \end{array}
\]
\end{minipage}%
\end{minipage}

\noindent\begin{minipage}{\textwidth}
\begin{minipage}[c][][c]{\dimexpr0.5\textwidth-0.5\Colsep\relax}
  \boxed{\textit{instantiate}(\Delta,\Delta)=\Delta}
  \begin{mathpar}
    \inferrule{\Delta_1' = \{x{:}\textit{inst-fun}(\tsym) \mid x{:}\tsym \in \Delta_1\} \\\\
               \Delta_2' = \{x{:}\tsym \mid x{:}\tsym\in\Delta_2, x\not\in\textit{dom}(\Delta_1)\}}{\textit{instantiate}(\Delta_1,\Delta_2)=\Delta_1'\cup \Delta_2'}
  \end{mathpar}
\end{minipage}\hfill
\begin{minipage}[c][][c]{\dimexpr0.5\textwidth-0.5\Colsep\relax}
  \boxed{\textit{inst-fun}(\tsym)=\tsym}
\begin{gather*}
    \textit{inst-fun}(\tsym_0,\ldots,\tsym_n{\to} \tsym) {=} \tsym_1,\ldots,\tsym_n {\to} \tsym\\
    \textit{inst-fun}(\tsym) = \tsym \text{ if }\tsym\neq \overline{\tsym_1}\to \tsym_2
  \end{gather*}
\end{minipage}%
\end{minipage}
  \caption{Translation for attribute reads and writes}
  \label{fig:translation-lookup}
\end{figure}
 
\begin{figure}
  \boxed{\flr{\tsym}=S}
  \begin{mathpar}
    \flr{\dyn}=\bdyn \and \flr{\mathsf{int}}=\mathsf{int} \and \flr{\overline{\tsym_1}\to \tsym_2}=\bto{|\overline{\tsym_1}|}\and
    \flr{\objty{X}{q}{\Delta}}=\bobjty{\flr{\Delta}}\and
    \flr{\classty{X}{q}{\Delta_1}{\Delta_2}{\overline{\tsym}}}=\bclassty{\flr{\Delta_1}}{|\overline{\tsym}|} \and\and \flr{\Delta}=\{x \mid x{:}\tsym\in \Delta\}
  \end{mathpar}
  \boxed{\tsym\lesssim \tsym}
  \begin{mathpar}
    \dyn\lesssim \tsym \and
    \tsym\lesssim \dyn \and
    \texttt{int}\lesssim\texttt{int}
    \and
    \inference{|\tsym_1| = |\tsym_3| \\ \overline{\tsym_3\lesssim \tsym_1} \\ \tsym_2 \lesssim \tsym_4}{\overline{\tsym_1}\to \tsym_2\lesssim\overline{\tsym_3}\to \tsym_4}\and
    \inference{\Delta_1\lesssim\Delta_2}{\objty{X}{q_1}{\Delta_1}\lesssim \objty{Y}{q_2}{\Delta_2}}\and
    \inference{\Delta_1\lesssim\Delta_3 \\ \Delta_2\lesssim\Delta_4 \\ |\tsym_1| = |\tsym_2| \\ \overline{\tsym_2\lesssim \tsym_1}}{\classty{X}{q_1}{\Delta_1}{\Delta_2}{\overline{\tsym_1}}\lesssim \classty{Y}{q_2}{\Delta_3}{\Delta_4}{\overline{\tsym_2}}}\and
    \inference{\Delta_1\lesssim\Delta_3}{\classty{X}{q_1}{\Delta_1}{\Delta_2}{\overline{\tsym}}\lesssim \objty{Y}{q_2}{\Delta_3}}\and
    \inference{\overline{\tsym_2\lesssim \tsym_1} \\ \objty{X}{q}{\textit{instantiate}(\Delta_1,\Delta_2)} \lesssim \tsym_3}{\classty{X}{q}{\Delta_1}{\Delta_2}{\overline{\tsym_1}}\lesssim \overline{\tsym_2}\to \tsym_3}
  \end{mathpar}
\noindent\begin{minipage}{\textwidth}
\begin{minipage}[c][][c]{\dimexpr0.5\textwidth-0.5\Colsep\relax}
  \boxed{\Delta\lesssim\Delta}\vspace{-2ex}
  \begin{mathpar}
\inference{\forall x \in \textit{dom}(\Delta_2).\;\Delta_1(x)\sim\Delta_2(x)}
          {\Delta_1 \lesssim \Delta_2}
  \end{mathpar}
\end{minipage}\hfill
\begin{minipage}[c][][c]{\dimexpr0.5\textwidth-0.5\Colsep\relax}
  \boxed{\Delta\sim\Delta}
  \begin{mathpar}
\inference{\forall x \in \textit{dom}(\Delta_1) \cap \textit{dom}(\Delta_2).\;\Delta_1(x)\sim\Delta_2(x)}
          {\Delta_1 \sim \Delta_2}
  \end{mathpar}
\end{minipage}%
\end{minipage}
  \caption{Relations used by the Anthill translation.
  Vertical bars denote the size of a collection.}
  \label{fig:relations}
\end{figure}

Figure \ref{fig:translation-lookup} shows the typechecking and
translation rules for object reads and writes, as well as the
\textit{mems} metafunction, which returns the attribute type $\Delta$
representing the attributes statically known to be present in the
value being read from or written to.

Rules \textsc{IGet} and \textsc{ISet} are used for reads and writes
from expressions that are statically known to contain the
attribute. In these cases, a check is inserted around the resulting
expression (\textsc{IGet}) or the expression to be written
(\textsc{ISet}) to make sure it corresponds with the type specified by
the type of the object. The type tag for this check is generated from
a static Anthill type using the $\flr{\tsym}$ metafunction, shown in
Figure \ref{fig:relations}. In the case of \textsc{ISet}, the type of
the term to be written $t_2$ must also be subtype-consistent (written
$\lesssim$) with the type of the attribute. This prevents static type
errors, such as attempting to write an \textsf{int} into a field that
expects an object.  Subtype-consistency, which can be thought of as
subtyping ``up to'' type dynamic \cite{Siek:2007qy}, is defined in
Figure \ref{fig:relations}, and indirectly uses the consistency
relation $\sim$ \cite{Siek:2006bh} from Appendix \ref{apx:full-semantics},
Figure \ref{fig:apx:more-metafunctions}.

The other rules, \textsc{IGet-Check} and \textsc{ISet-Check} are used
when it is not statically known that the subject of a read or write
contains the appropriate attribute, either because it is of type
$\dyn$, or because it has a class or object type that omits the
attribute. In these cases, the type system only allows the read or
write if the object's type is \emph{open} (denoted $\lozenge$) --- the
definition of the \textit{queryable} metafunction means that a $\star$-typed
object is always $\lozenge$. Width downcasts are implicitly performed on open types but not on
closed types ($\blacklozenge$). \textsc{IGet-Check} inserts a check to
ensure that the attribute is present, but \textsc{ISet-Check} only
checks that the subject of the write is an object (i.e.\ not an
integer or function) because attribute writes can be used to add new
attributes to objects and classes.

\subsubsection{Functions and applications}

\begin{figure}
\boxed{\cig\vdash t\leadsto e:\tsym}
  \begin{mathpar}
\fninference{IFun}{\Gamma,\overline{x:\tsym_1} \vdash t\leadsto e:\tsym_2' \\ \tsym_2' \lesssim \tsym_2}
            {\Gamma {\vdash} (\funt{\overline{x{:}\tsym_1}}{\tsym_2}{t}) \leadsto (\fune{\overline{x}}{\overline{\mathtt{let}~x=\tcheck{x}{\flr{\tsym_1}}~\mathtt{in}~}e}):\overline{\tsym_1}{\to} \tsym_2}
\and
\ninference{IApp-Dyn}{\cig\vdash t_1\leadsto e_1:\dyn \\ \cig\vdash \overline{t_2\leadsto e_2:\tsym}}
          {\cig\vdash t_1(\overline{t_2}) \leadsto (\tcheck{e_1}{\bto{|\overline{\tsym}|}})(\overline{e_2}) : \dyn}
\and
\ninference{IApp-Fun}{\cig\vdash t_1\leadsto e_1:\overline{\tsym_1} \to \tsym_2 \\ 
           \Gamma\vdash \overline{t_2\leadsto e_2:\tsym_1'} \\\\ |\overline{\tsym_1}| = |\overline{t_2}| \\ 
           \overline{\tsym_1' \lesssim \tsym_1}}
          {\cig\vdash t_1(\overline{t_2})\leadsto \tcheck{(e_1(\overline{e_2}))}{\flr{\tsym_2}} : \tsym_2}
\and
\ninference{IApp-Constr}{\cig\vdash t_1\leadsto e_1:\classty{X}{q}{\Delta_1}{\Delta_2}{\overline{\tsym_1}} \\
           \Gamma\vdash \overline{t_2\leadsto e_2:\tsym_1'} \\ |\overline{\tsym_1}| = |\overline{t_2}| \\
           \overline{\tsym_1' \lesssim \tsym_1} \\ \tsym_2 = \objty{X}{q}{\textit{instantiate}(\Delta_1,\Delta_2)}}
          {\cig \vdash t_1(\overline{t_2}) \leadsto \tcheck{(e_1(\overline{e_2}))}{\flr{\tsym_2}}:\tsym_2}
  \end{mathpar}
  \caption{Translation for functions and applications}
  \label{fig:translation-fun}
\end{figure}

\begin{figure*}
\boxed{\cig\vdash t\leadsto e:\tsym}
\begin{mathpar}
  \ninference{IClass}{\cig\vdash \overline{t_s\leadsto e_s:\tsym_s} \\
           \tsym_{class}=\classty{X}{q}{\Delta_1}{\Delta_2}{\overline{\tsym_c}}\\
\cig\vdash_\sigma c \leadsto e_c : \overline{\tsym_c} \\ \Gamma; \tsym_{class}\vdash_\varsigma \overline{m\leadsto e_m:\tsym_m} \\ \cig\vdash \overline{t_f\leadsto e_f:\tsym_f} \\
           \overline{e_s'=\tcheck{e_s}{\bclassty{\flr{\textit{mems}(\tsym_s)}}{\mathsf{Any}}}} \\ 
           \forall x \in \mathit{dom}(\Delta_1),~ (\overline{\ell_f}\times \overline{A_f} \cup 
           \overline{\ell_m}\times \overline{A_m} \cup \overline{\mathit{mems}(A_s)})(x) \lesssim \Delta_1(x)}
          {
              \cig\vdash \classt{X}{q}{\Delta_1}{\Delta_2}{\overline{t_s}}{\overline{\ell_f=^M m}}{\overline{\ell_m=^F t_m}}{c} 
              \leadsto \\\\ \classe{X}{\overline{e_s'}}{\overline{\ell_m=e_m},\overline{\ell_f=e_f}}{e_c}{}:\tsym_{class}
            }
\end{mathpar}
\boxed{\cig\vdash_\sigma c\leadsto e:\overline{\tsym}}
\begin{mathpar}
\fninference{IConstruct}{\Gamma,x_s{:}\dyn,\overline{x{:}\tsym_1}\vdash t\leadsto e:\tsym_2}
          {\cig\vdash_\sigma \constructt{\overline{x{:}\tsym_1}}{t} \leadsto \fune{x_s,\overline{x}}{\overline{\mathtt{let}~x=\tcheck{x}{\flr{\tsym_1}}~\mathtt{in}~}e}:\overline{\tsym_1}}
\end{mathpar}
\boxed{\Gamma; \tsym\vdash_\varsigma m\leadsto e:\tsym}
\begin{mathpar}
\ninference{IMethod}{\tsym_o = \objty{X}{q}{\textit{instantiate}(\Delta_1,\Delta_2)} \\ \tsym_2' \lesssim \tsym_2\\ \Gamma,x_s{:}\tsym_o,\overline{x{:}\tsym_1}\vdash t\leadsto e:\tsym_2'}
          {
              \Gamma; \classty{X}{q_1}{\Delta_1}{\Delta_2}{\overline{\tsym_c}}\vdash_\varsigma \method{\overline{x{:}\tsym_1}}{\tsym_2}{t} \leadsto \\\\
              \fune{x_s,\overline{x}}{\lett{x_s}{\tcheck{x_s}{\flr{\tsym_o}}}{\overline{\mathtt{let}~x=\tcheck{x}{\flr{\tsym_1}}~\mathtt{in}~}e}:\overline{\tsym_1}\to \tsym_2}
            }
\end{mathpar}
  \caption{Translation for classes and methods}
  \label{fig:translation-classes}
\end{figure*}

The translation from Anthill to \upython{} for functions and
applications is shown in Figure \ref{fig:translation-fun}. Anthill
functions are translated into \upython{} functions using the
\textsc{IFun} rule, which inserts checks to ensure that parameters
have their expected types in the body of the function. Functions are
not responsible for ensuring that they return results that correspond
to their return types (modulo static type errors). Instead, callers
check that the resultant value is of the appropriate type, as
indicated by \textsc{IApp-Fun} and \textsc{IApp-Constr}. This is
especially important for constructor calls, which operate by mutating
a dynamically typed empty receiver (as discussed in Section
\ref{sec:statics:dynamic} and below in Section
\ref{sec:fml:trans:class}). The caller is responsible for checking
that the resulting object has the type of an instance of the class.

\subsubsection{Classes and objects}
\label{sec:fml:trans:class}

The rule for translating Anthill Python class definitions into
\upython{} is given as \textsc{IClass} in Figure
\ref{fig:translation-classes}. In this translation, the superclasses
of the class are translated and have checks placed around them to
ensure that they are, in fact, classes --- an error occurs if a
class declaration inherits from a non-class value.

Class fields, methods, and the constructor are also translated into
\upython{} expressions, and the latter two cases require new
judgments. Constructors and methods are similar to functions, using
their first argument as the receiver.  In \textsc{IConstruct}, the
receiver $x_s$ has dynamic type, because the instance has yet to
actually be constructed. No concern is paid to the return type because
the return value is discared when called. Methods, translated by
\textsc{IMethod}, are similar, except that the receiver is given the
type of an instance of the enclosing class. In both cases, as in
\textsc{IFun}, checks are inserted to ensure that the parameters have
the correct type.

\textsc{IClass} also ensures that all declared class attributes (those
in $\Delta_1$) can be found in the fields, methods, or superclasses of
the class, and that their types are subtype-consistent with their
declared types.

As observed by \citet{Takikawa:2012ly}, width subtyping for class
types is statically unsound, and because Anthill's object types
support width subtyping, classes do as well, as shown in Figure
\ref{fig:relations}. Anthill's static type system does, therefore,
admit some maximally annotated programs that result in runtime check
failures. (Since the self-reference of a constructor is always
dynamically typed, programs with classes are never \emph{fully}
annotated anyway.) However, the pervasive nature of \transient{}'s
runtime checking means that this cannot lead to uncaught type errors
at runtime.

\subsubsection{Eager and delayed error detection}

The translation from Anthill to \upython{} inserts checks only when
necessary to ensure soundness, whereas Reticulated performs extra
eager checks to provide better feedback to programmers. Eagerly
checking objects does not aid in ensuring soundness, however, and
therefore the translation from Anthill to \upython{} installs only
shallow checks at use sites. For example, in the following program, a
class specifies that its instances should have the field
\li{x:int}. However, its constructor instead writes a function to
\li{x}.

\begin{minipage}{1.0\linewidth}
  \begin{lstlisting}
let C = class ()(*$\llparenthesis \lozenge, \emptyset, \text{x:int}\rrparenthesis$\label{ex:ant:classdef}*):
           init=((*$\sigma$*) self. self.x=(*$\lambda$*) z. z)
in let c = C() (*\label{ex:ant:constr}*)
in c.x (*\label{ex:ant:get}*)
  \end{lstlisting}
\end{minipage}

Reticulated would detect this error at the object instantiation site
at line \ref{ex:ant:constr}, because \li{c.x} is a function, not an
int. However, even if the constructor was corrected, \li{c.x} could
still become a non-\li{int}, if some untyped reference to \li{c}
strongly updated \li{x}. In that case, fully checking \li{c} at line
\ref{ex:ant:constr} would have accomplished nothing in terms of
ensuring soundness --- it still needs to be checked again at every
use site. In Anthill, the error in the above program is detected at
the use site on line \ref{ex:ant:get}, when an inserted check attempts
to verify that the result of \li{c.x} is an \li{int}.

\subsection{Dynamic semantics of \upython}

\begin{figure}[bp]%{l}{.7\textwidth}
  \centering
  \[
  \begin{array}{lcl}
    v & ::= & n \mid a \mid \fune{\overline{x}}{e} \\
    E & ::= & \chole \mid E(\overline{e}) \mid v(\overline{v},E,\overline{e}) \mid E.\ell \mid E.\ell=e 
    \mid  v.\ell=E \mid \lett{x}{E}{e} \mid \tcheck{E}{S}\mid\\
    && \classe{X}{\overline{v},E,\overline{e}}{\overline{\ell=e}}{e}{} \mid  \classe{X}{\overline{v}}{\overline{\ell=e}}{E}{} \mid \\ && 
    \classe{X}{\overline{v}}{\overline{\ell=v},\ell=E,\overline{\ell=e}}{v}{} \\
    \mu & ::= & \overline{a\mapsto h}\\
    h & ::= & \classh{\overline{a}}{M}{e} \mid \objh{a}{M}\\
    M & ::= & \overline{\ell=v} \\
    r & ::= & (e\mid\mu) \mid \mathtt{casterror} \mid \mathtt{pyerror}
  \end{array}
  \]
  \caption{Runtime syntax for \upython}
  \label{fig:upython-runtime-syntax}
\end{figure}

The runtime behavior of \upython{} is defined in terms of a
single-step reduction semantics using evaluation contexts, shown in
Figure \ref{fig:eval}. The runtime structures (values, heaps, etc.)
are defined in Figure \ref{fig:upython-runtime-syntax}.  The reduction
relation $e\mid\mu\longrightarrow r$ steps from a pair of an
expression $e$ and a heap $\mu$ to some result $r$, which is either an
error or a pair of an expression and a heap. Values are numbers,
functions, and heap addresses. Heaps contain mappings from addresses
$a$ to heap values $h$, which are either classes or objects. Class
heap values contain references to all superclasses and all of the
class' attributes, and objects contain instance variables and a
reference to the object's class. Error results come in two varieties:
\texttt{casterror}, which is the result of a check $\tcheck{v}{S}$
that fails, or \texttt{pyerror}, which is the result of a dynamic type
error (such as, for example, calling a number as though it were a
function). We choose to have such cases reduce to an actual result,
rather than treating them as stuck states, because \upython{} is a
dynamically typed language which contains not just programs translated
from Anthill, but other programs which produce runtime errors; a
\upython{} program that reduces to \texttt{pyerror} is the equivalent
of a Python program that raises a \li{TypeError} exception. As rules
\textsc{EPyError} and \textsc{ECastError} in Figure \ref{fig:eval}
show, both kinds of errors are propagated upwards by the context
rules.

Rules \textsc{ECheck1} and \textsc{ECheck2} show \upython's reduction
rules for \transient{}'s type checks, which use the \textit{check}
metafunction defined in Figure \ref{fig:eval}. In all cases, the
reduction results in either a \texttt{casterror} or the checked value
itself. The cases where the checked value is an address are the most
interesting: class values can be called like functions and read
from/written to like objects, so the check can be successful with both
function and object type tags. The \textit{hasattrs} metafunction,
used for verifying if all attributes in $\delta$ are reachable from a
heap value, and \textit{param-match}, used for checking the arity of a
callable value, are defined in Appendix \ref{apx:full-semantics},
Figure \ref{fig:apx:more-metafunctions}.

Most remaining reduction rules for \upython{} are standard, except
that every case that would be stuck in a statically typed calculus has
a reduction to \texttt{pyerror}. Application requires some work: the
value being called may be a class, in which case an empty object is
created and passed to the object's constructor (\textsc{EApp2}).
Method lookup on objects, shown in rule \textsc{EGet1} and the
\textit{lookup} metafunction, is complex: methods have to be looked up
from a class, curried, and given the receiver as the first
argument. The \textit{getattr} partial function, which traverses the
inheritance hierarchy of its argument to find the definition site of
an attribute, is defined in Appendix \ref{apx:full-semantics}, Figure
\ref{fig:apx:more-metafunctions}.

One quirk to the \textit{lookup} metafunction is that attempting to
access a class' nullary method from an instance leads directly to a
\texttt{casterror}, rather than a \texttt{pyerror}. In Python, method
arities are reduced by one when called from an instance, because the
instance is bound to the first parameter as the receiver. A nullary
method bound like this essentially expects $-1$ arguments: calling it
with zero arguments results in a type error, because too many
arguments were provided. In Reticulated, a transient check on such a
function always fails, but rather than providing uncallable,
``negative-ary'' functions in Anthill, such an evaluation simply
results in \texttt{casterror}.

The semantics of \upython do not refer to casts, types, or other
features related to translation from Anthill, except in the type check
rules and the method access rule (as mentioned above) --- everything
else is as expected for an untyped language. This supports our claim
that \upython{} models Python, and that we do not have to modify the
underlying semantics of Python itself in order to implement
Reticulated with the \transient{} semantics, nor do we require a
complicated implementation of proxies capable of handling Python's
features, as does \guarded{}. The complex parts of \upython's
semantics implement features like method binding, and are not affected
by the existence of runtime type checks.

\begin{figure}[tbp]
\boxed{e\mid\mu\longrightarrow r}
  \begin{mathpar}
    \ninference{EStep}{e\mid\mu\longrightarrow e'\mid\mu'}
              {E[e]\mid\mu \longmapsto E[e']\mid\mu'}
\and
    \ninference{EPyError}{e\mid\mu\longrightarrow \mathtt{pyerror}}
              {E[e]\mid\mu \longmapsto \mathtt{pyerror}}
\and
    \ninference{ECastError}{e\mid\mu\longrightarrow \mathtt{casterror}}
              {E[e]\mid\mu \longmapsto \mathtt{casterror}}
  \end{mathpar}
\newcommand{\steplab}[1]{\scriptsize\textsc{(#1)}}
\[
\begin{array}{lcl}
%  \multicolumn{4}{l}{\textsc{Checks}}\\
   \steplab{ECheck1,2}\hfill\tcheck{v}{S}\mid\mu \hfill& \longrightarrow & \left\{
      \begin{array}{p{2.2cm}l}
        $v \mid \mu$ & \text{if }\mathit{check}(v,\mu,S)\\[.5ex]
        $\mathtt{casterror}$ & \text{otherwise}
      \end{array}\right.\\
 \steplab{EApp1,2,3}\hfill v_1(\overline{v_2})\mid\mu & \longrightarrow & \left\{
    \begin{array}{p{2.2cm}l}
      $e[\overline{x/v_2}] \mid \mu$ & \text{if }v_1=\lambda \overline{x}.e\\
      &\text{and } |\overline{x}| =|\overline{v_2}|\\[.5ex]
      $\mathtt{let}~\_=$ & \text{if } v_1=a \\
      \quad$v'(a',\overline{v_2})$ & \text{and } \mu(a)=\classh{\overline{a''}}{M}{v'}\\
      \quad$\mathtt{in}~a' \mid \mu'$& \text{and } a' \text{ fresh}\\
      & \text{and } \mu'=\mu[a'\mapsto \objh{a}{\emptyset}]\\[.5ex]
      $\mathtt{pyerror}$ & \text{otherwise}
    \end{array}\right.\\[.5ex]
  \steplab{ELet}\hfill\lett{x}{v}{e} & \longrightarrow & \quad\;e[x/v]\\[.5ex]
  \overset{\steplab{EClass1,2}\hfill}{\classe{X}{\overline{a}}{M}{v}{} \mid \mu} & \longrightarrow &  \left\{
    \begin{array}{p{2.2cm}l}
      $a'\mid\mu[a'\mapsto h]$ & \text{if } \overline{\mu(a)=\classh{\overline{a''}}{M'}{v'}} \\
      & \text{and } \textit{param-match}(v,\mu,\mathsf{Any})\\
      & \text{and } h=\classh{\overline{a}}{M}{v}\\
      & \text{and } a' \text{ fresh}\\[.5ex]
      $\mathtt{pyerror}$ & \text{otherwise}
    \end{array}\right.\\
  \overset{\steplab{EClass3}\hfill}{\classe{X}{\overline{v_1}}{M}{v_2}{} \mid \mu} & \longrightarrow &
  \begin{array}{ll}
    \;\;\;\mathtt{pyerror} & \quad\;\text{if } \overline{v_1 \neq a}
  \end{array}\\[.5ex]
  \steplab{EGet1,2}\hfill a.\ell\mid\mu & \longrightarrow & \left\{
    \begin{array}{p{2.2cm}l}
      $r$ & \text{if } \mathit{lookup}(a,\mu(a),\ell,\mu)=r\\
      $\mathtt{pyerror}$ & \text{otherwise}
    \end{array}\right.\\
  \steplab{EGet3}\hfill v.\ell\mid\mu & \longrightarrow &
  \begin{array}{ll}
    \;\;\;\mathtt{pyerror} & \quad\;\text{if } v\neq a
  \end{array}\\
   \steplab{ESet1,2,3}\hfill a.\ell=v\mid\mu & \longrightarrow & \left\{
    \begin{array}{p{2.2cm}l}
      $0\mid\mu[a\mapsto h']$ & \text{if } \mu(a)=\objh{a'}{M} \\
      & \text{and } h'=\objh{a'}{M[\ell=v]}\\[.5ex]
      $0 \mid \mu[a\mapsto h']$ & \text{if } \mu(a)=\classh{\overline{a'}}{M}{v'}\\
      & \text{and } h'=\classh{\overline{a'}}{M[\ell=v]}{v'}\\[.5ex]
      $\mathtt{pyerror}$ & \text{otherwise}
    \end{array}\right.\\
  \steplab{ESet4}\hfill v_1.\ell=v_2\mid\mu & \longrightarrow &
  \begin{array}{ll}
    \;\;\;\mathtt{pyerror} & \quad\;\text{if } v_1\neq a
  \end{array}\\
\end{array}
\]
\boxed{e\mid\mu\longrightarrow^{*}r}
\begin{mathpar}
  \fninference{MRefl}{ }{e\mid\mu\longrightarrow^{*}e\mid\mu}\and
  \fninference{MPyErr}{e\mid\mu\longmapsto\mathtt{pyerror}}{e\mid\mu\longrightarrow^{*}\mathtt{pyerror}}\and
  \fninference{MCastErr}{e\mid\mu\longmapsto\mathtt{casterror}}{e\mid\mu\longrightarrow^{*}\mathtt{casterror}}\and
  \fninference{MChain}{e\mid\mu\longmapsto e'\mid\mu' \\ e'\mid\mu'\longrightarrow^{*}r}{e\mid\mu\longrightarrow^{*}r}
\end{mathpar}
  \caption{\upython{} evaluation rules}
  \label{fig:eval}
\end{figure}

\begin{figure}[tbp]
\boxed{\mathit{check}(v,\mu,S)}
  \begin{mathpar}
    \inferrule{ }{\mathit{check}(v,\mu,\bdyn)}\and
    \inferrule{ }{\mathit{check}(n,\mu,\mathsf{int})}\and
    \inferrule{\mathit{hasattrs}(a,\delta,\mu)}{\mathit{check}(a,\mu,\bobjty{\delta})}\and
    \inferrule{|\overline{x}|=n}{\mathit{check}(\lambda \overline{x}.e, \mu, \bto{n})}\and
    \inferrule{\mu(a)=\classh{\overline{a'}}{M}{v} \\\\ \textit{param-match}(v,\mu,n+1)}{\mathit{check}(a, \mu, \bto{n})}\and
    \inferrule{\mu(a)=\classh{\overline{a'}}{M}{v} \\\\\textit{param-match}(v,\mu,C) \\ \mathit{hasattrs}(a,\delta,\mu)}{\mathit{check}(a,\mu,\bclassty{\delta}{C})}
  \end{mathpar}
\boxed{\mathit{lookup}(a,h,\ell,\mu)=r}
  \begin{mathpar}
    \inferrule{M(\ell)=v}{\mathit{lookup}(a,\objh{a'}{M},\ell,\mu)=v\mid\mu}\and
    \inferrule{\textit{getattr}(a,\ell,\mu)=v}{\mathit{lookup}(a,\classh{\overline{a'}}{M}{v'},\ell,\mu)=v\mid\mu}\and
    \inferrule{\ell\not\in\mathit{dom}(M) \\\\ \mathit{getattr}(a,\ell,\mu)=v \\ v\neq\lambda\overline{\ell}.e}{\mathit{lookup}(a,\objh{a'}{M},\ell,\mu)=v\mid\mu}\and
    \inferrule{\ell\not\in\mathit{dom}(M)\\\\\mathit{getattr}(a,\ell,\mu)=v \\ v=\lambda\overline{x}.e \\ |\overline{y}|=|\overline{x}|-1}{\mathit{lookup}(a,\objh{a'}{M},\ell,\mu)=\lambda\overline{y}.v(a,\overline{y})\mid\mu}\and
    \inferrule{\ell\not\in\mathit{dom}(M) \\ \mathit{getattr}(a,\ell,\mu)=v \\ v=\lambda\overline{x}.e \\ |\overline{x}| =0}{\mathit{lookup}(a,\objh{a'}{M},\ell,\mu)=\mathtt{casterror}}
  \end{mathpar}
\caption{Relations used in \upython{} evaluation}
\label{fig:lookup}
\end{figure}

\section{Open World Soundness of Anthill Python}
\label{sec:proofs}

With static and dynamic semantics for Anthill Python, we now wish to
show that Anthill admits the property of \emph{open world soundness.}
That is, an Anthill program, translated into \upython, can interact
with other \upython{} code without the translated Anthill code causing
any \texttt{pyerror}s. However, in our presentation so far \upython{}
does not distinguish between translated and untranslated code, so we
begin by introducing \emph{origin tracking} to \upython{}
expressions. This indicates whether an expression originated in
typed or untyped code.

Since \upython{} is a dynamically typed language, there is nothing
preventing programs from resulting in runtime errors. For example,
consider the following \upython{} program, in which a function \li{f}
calls its argument on the integer $42$.

\begin{lstlisting}
let f = ((*$\lambda$*) v. v(42)) in (*\label{ex:ow:pe}*)
f(21) (*\label{ex:ow:bc}*)
\end{lstlisting}
Since \li{f} is passed $21$ at line \ref{ex:ow:bc}, this program will
result in a \texttt{pyerror} at the call site in line \ref{ex:ow:pe}
by rule \textsc{EApp3}. Suppose instead that \li{f} was typed Anthill
Python code rather than \upython{}:

\begin{lstlisting}
let f = ((*$\lambda$*) v:Callable([int], int). v(42)) in
...
\end{lstlisting}
This version specifies that \li{f}'s parameter be a function. It can
then be translated into \upython{} and used in place of the original \li{f}:

\begin{lstlisting}
let f = ((*\makebox[0pt][l]{\color{\cicolor}\rule[-1ex]{10.5em}{2.6ex}}$\lambda$*) v. let v = v(*$\Downarrow_{1\to}$*) in (*\label{ex:owt:ce}*)
                (*\makebox[0pt][l]{\color{\cicolor}\rule[-0.5ex]{5.1em}{2.4ex}}*)(v(42))(*$\Downarrow_{\text{int}}$*)) in 
f(21)
\end{lstlisting}
The result is a \upython{} program that is partially translated, typed
Anthill code (highlighted in \colorbox{\cicolor}{yellow}) and partially code that originates in \upython{}. In this
case, the result of the program is no longer \texttt{pyerror}, but a
\texttt{casterror} at line \ref{ex:owt:ce}, which indicates that
\li{f} has been passed something of the wrong type. Open world
soundness guarantees that mixed programs like this can only result in
\texttt{pyerror} in those sections that originated in \upython{}.

To reason about such mixed programs, we describe portions of a program
as originating in either Anthill or \upython. Consider again the
example from above, but with origin made explicit:

\begin{lstlisting}
let f = ((*\makebox[0pt][l]{\color{\cicolor}\rule[-1ex]{10.5em}{2.6ex}}$\lambda$*) v. let v = v(*$\Downarrow_{1\to}$*) in
                (*\makebox[0pt][l]{\color{\cicolor}\rule[-0.5ex]{5.6em}{2.4ex}}*)(v(42)(*$^\circ$*))(*$\Downarrow_{\text{int}}$*)) in (*\label{ex:owl:ty}*)
f(21)(*$^\bullet$*)(*\label{ex:owl:un}*)
\end{lstlisting}

The call site at \ref{ex:owl:ty} is labeled $\circ$ to indicate that
it originated Anthill, and should never result in
\texttt{pyerror}. The call site at line \ref{ex:owl:un}, on the other
hand, is labeled $\bullet$ to indicate that it originated in \upython,
and \emph{can} result in \texttt{pyerror}.

In the remainder of this section, we further develop this notion of
origin, and use it in defining a type system for \upython, which
enforces restrictions on code that originated in Anthill while being
permissive for code that does not. We then use this type system to
prove open world soundness.

\subsection{Expression origin}

\newcommand{\cc}{\mathcal{C}}

\begin{figure}
  \centering
  \[
  \begin{array}{lcl}
    p & ::= & \circ \mid \bullet \\
    e & ::= & \ldots \mid e(\overline{e})^p \mid e.\ell^p \mid e.\ell^p=e \mid  \classe{X}{\overline{e}}{\overline{\ell=e}}{e}{p}\\ 
    r & ::= & \ldots \mid \mathtt{pyerror}(p)\\
    \Sigma & ::= & \overline{a{:}S}\\
    \cc & ::= & \chole \mid \cc(\overline{e})^\bullet \mid e(\overline{e},\cc,\overline{e})^\bullet \mid \cc.\ell^\bullet \mid \cc.\ell^\bullet=e 
    \mid e.\ell^\bullet=\cc \mid \classe{X}{\overline{e},\cc,\overline{e}}{\overline{\ell=e}}{e}{\bullet} \mid \\
    && \classe{X}{\overline{e}}{\overline{\ell=e}}{\cc}{\bullet} \mid  \classe{X}{\overline{v}}{\overline{\ell=e},\ell=\cc,\overline{\ell=e}}{e}{\bullet} \mid \\
    && \lett{x}{\cc}{e} \mid \lett{x}{e}{\cc} \mid \tcheck{\cc}{S} \mid 
    \lambda \overline{x}.\;\cc
  \end{array}
  \]
  \caption{Labeled syntax and contexts for \upython}
  \label{fig:upython-labeled-syntax}
\end{figure}

Figure \ref{fig:upython-labeled-syntax} shows the syntax of
\upython{} extended with origin labels $p$. In this system,
there are only two labels: $\circ$, for code translated from
well-typed Anthill Python programs, and $\bullet$, for
code that originates in \upython{} and is not statically typed.

Unlike typical blame labels in contracts and gradual typing, which
attach to casts or contract monitors \cite{Ahmed:2011fk,
  Findler:2002eu, Takikawa:2012ly}, these labels are attached to
expressions, and specifically to elimination forms that can result in
\texttt{pyerror}. (No labels are attached to, for example,
let-binding, because let-binding cannot on its own lead to a
\texttt{pyerror}.)  We also change the \texttt{pyerror} result so that
it reports the origin of the expression that triggered the error. Code
translated from Anthill is \emph{always} labeled with $\circ$, while
\upython{} programs entirely labeled with $\bullet$ represent untyped
Python code. The version of \upython's dynamic semantics that includes
origin is shown in Appendix \ref{apx:full-semantics}, Figure
\ref{fig:apx:eval}; it differs from that shown in Figure
\ref{fig:eval} only in its use and propagation of labels.

\subsection{Applying types to \upython}

In addition to indicating whether a dynamic type error occurred in
typed or untyped code, origin labels also let us define a type system
for \upython{} in order to state and prove open world soundness. This
type system relates expressions to type tags $S$, as defined in Figure
\ref{fig:upython-syntax}, which are repurposed as types.  An
illustrative excerpt of this type system is shown in Figure
\ref{fig:typing-excerpt-upython}. Two rules are provided for each
labeled expression, one for $\circ$ and one for $\bullet$. The
$\bullet$ rule requires that all subexpressions be typed as \bdyn,
which is the top of the subtyping hierarchy for $S$ as shown in Figure
\ref{fig:relations}, so no programs can be rejected unless they
contain $\circ$-labeled expressions. The $\circ$ rules are more
restrictive --- a $\circ$-labeled expression is well-typed when it
cannot step to \texttt{pyerror}, and so $\circ$ rules require that
subexpressions have the specific types necessary to ensure that. All
expressions of either variety, other than checks and introduction
forms, are typed as \bdyn. This ensures runtime type checks exist in
$\circ$-labeled code, because type checks are the only way to obtain
expressions with precise types like \textsf{int} and $\bto{n}$ (other
than introduction forms like numbers or lambdas).

In rule \textsc{TCheck}, the type of a check expression is its type
tag. In \textsc{TApp} we require that, in a call labeled with
$\circ$, the callee has a function type. The only way for it to have
a function type is if it is a bare lambda, if it is a variable bound
by \textsc{TLet}, or if it is some other expression nested within a
check (typed by \textsc{TCheck}). By contrast,
\textsc{TApp-Dyn} places no requirements on the types of
subexpressions of an untyped, $\bullet$-labeled application. Even if a
subexpression has a more specific type, \bdyn~is the top of the
subtyping hierarchy, so by \textsc{TSubsump} any well-typed
expression can appear in function position. Therefore, an obviously
ill-typed program like \li{4(2)} (calling $4$ as though it were a
function) will be ill-typed if the call is labeled with $\circ$, but
allowed if it is labeled with $\bullet$.  The full type system for
\upython{} with labels is shown in Appendix \ref{apx:full-semantics},
Figure \ref{fig:apx:upython-blame-typing}.

\begin{figure}
\boxed{\tlg\vdash e:S}
\begin{mathpar}
\ninference{TSubsump}{\tlg\vdash e:S_2 \\ S_2 <: S_1}{\tlg\vdash e:S_1}\and
\ninference{TCheck}{\tlg\vdash e:\bdyn}{\tlg\vdash \tcheck{e}{S} : S}\and
\ninference{TLet}{\tlg\vdash e_1:S_1 \\ \Gamma,x{:}S_1{\mid}\Sigma\vdash e_2:S_2}{\tlg\vdash \lett{x}{e_1}{e_2} : S_2}\and
\ninference{TApp}{\tlg\vdash e_1:\bto{n} \\ \tlg\vdash \overline{e_2:\bdyn} \\ |\overline{e_2}| =n}{\tlg\vdash e_1(\overline{e_2})^\circ : \bdyn}\and
\ninference{TApp-Dyn}{\tlg\vdash e_1:\bdyn \\ \tlg\vdash \overline{e_2:\bdyn}}{\tlg\vdash e_1(\overline{e_2})^\bullet : \bdyn}
\end{mathpar}
\boxed{S <: S}\vspace{-2ex}
  \begin{mathpar}
    S <: \bdyn \and \texttt{int} <: \texttt{int}\and
    %% The following should be proved as lemmas! -Jeremy
 %% \inference{ }{S <: S} \and
 \inference{S_1 <: S_2 \\ S_2 <: S_3}{S_1 <: S_3}
\and
    \inference{\delta_2 \subseteq \delta_1}{\bobjty{\delta_1} <: \bobjty{\delta_2}}\and
    \inference{ }{\bclassty{\delta}{n} <: \bclassty{\delta}{\mathsf{Any}}}\and
    \inference{\delta_2 \subseteq \delta_1}{\bclassty{\delta_1}{C} <: \bclassty{\delta_2}{C}}\and
    \inference{ }{\bclassty{\delta}{C} <: \bobjty{\delta}}\and
    \inference{ }{\bclassty{\delta}{n} <: \bto{n}}
  \end{mathpar}
  \caption{Excerpt of type system for typed expressions in \upython}
  \label{fig:typing-excerpt-upython}
\end{figure}

\subsection{Code contexts allow embedding typed code in untyped}

To reason about Anthill-translated code interacting with other
\upython{} code, we use code contexts
$\cc$\cite{Harper:2012aa}. Contexts are defined in Figure
\ref{fig:upython-labeled-syntax}; in this work we are concerned with
embedding typed code in untyped contexts, so these contexts are
$\bullet$-labeled. Code contexts are typed using the judgment
$\cc:\Gamma{;}S_1\Rightarrow \Gamma'{;}S_2$, where if the hole in a
context $\cc$ is filled by an expression of type $S_1$ under $\Gamma$,
then the result is an expression of type $S_2$ under $\Gamma'$. We write $\cc[e]$ for
the composition of a context and an expression, which is itself an
expression. For example, $\chole:\Gamma{;}S\Rightarrow \Gamma{;}S$,
because a hole filled by an expression is just that expression.  The
full rules for typing contexts are given in Appendix
\ref{apx:full-semantics}, Figure \ref{fig:apx:context-typing}.

\subsection{Open world soundness}

Armed with notions of origination and a type system for \upython, we
can now describe the proof of open world soundness and its key lemmas.

First, if an Anthill term $t$ has type $A$ under $\Gamma$ and is
translated into a \upython{} expression $e$, then $e$ has type
$\flr{A}$ under $\flr{\Gamma}$.\footnote{Here, $\flr{\Gamma}$ is the
  result of applying $\flr{\tsym}$ (see Figure \ref{fig:relations}) to
  all the types in $\Gamma$.}

\begin{lemma}[Anthill translation preserves typing]\hfill
  \begin{itemize}
  \item If $\Gamma\vdash t\leadsto e:\tsym$, then $\flr{\Gamma};\emptyset\vdash
  e:\lfloor \tsym\rfloor$.
\item If $\Gamma\vdash c\leadsto e: \overline{\tsym}$, then
  $\flr{\Gamma};\emptyset\vdash e:\bto{|\overline{\tsym}|}$.
\item If $\Gamma{;} \tsym_1 \vdash d \leadsto e:\tsym_2$, then
  $\flr{\Gamma};\emptyset\vdash e:\flr{\tsym_2}$
  \end{itemize}
\end{lemma}

If $e$ is plugged in a context
$\cc:\flr{\Gamma};\flr{A}\Rightarrow\emptyset;S$, then the resulting
term $\cc[e]$ has type $S$.

\begin{lemma}[Composition]
  If $\cc:\Gamma;S\Rightarrow \Gamma';S'$ and
  $\Gamma;\emptyset\vdash e:S$, then $\Gamma';\emptyset\vdash \cc[e]:S'$.
\end{lemma}

Finally, when $\cc[e]$ is evaluated, if it
does not diverge it will result in a value of type $S$, a
\texttt{casterror}, or a \texttt{pyerror} that points to untyped code
as the source of the error (i.e. $e$ is not responsible).

\begin{lemma}[Preservation]
  If $\emptyset;\Sigma \vdash e:S$, $\Sigma\vdash\mu$, and
  $e\mid\mu \longrightarrow e'\mid\mu'$, then
  $\emptyset;\Sigma'\vdash e':S$ and $\Sigma'\vdash\mu'$ and
  $\Sigma \sqsubseteq \Sigma'$.
\end{lemma}

\begin{lemma}[Progress with no typed pyerrors]
  If $\emptyset;\Sigma \vdash e:S$ and $\Sigma\vdash\mu$, then either $e$ is a value or $e\mid\mu\longrightarrow r$ and either:
  \begin{itemize}
  \item $r=\mathtt{pyerror}(\bullet)$, or
  \item $r=\mathtt{casterror}$, or
  \item $r=e'\mid\mu'$.
  \end{itemize}
\end{lemma}

These key lemmas let us prove the overall statement of open world soundness:

\begin{theorem}[Open world soundness]
  If $\Gamma \vdash t \leadsto e : \tsym$, then for any context $\cc$
  such that $\cc:\flr{\Gamma}{;}\flr{\tsym}\Rightarrow \emptyset{;}S$,
  we have that $\emptyset;\emptyset\vdash \cc[e]:S$ and either:
  \begin{itemize}
  \item for some $v,\Sigma,\mu$, $\cc[e]\mid\emptyset\longrightarrow^{*} v\mid\mu$ and $\emptyset;\Sigma\vdash v:S$ and $\Sigma\vdash\mu$, or
  \item $\cc[e]\mid\emptyset\longrightarrow^{*}\mathtt{pyerror}(\bullet)$, or
  \item $\cc[e]\mid\emptyset\longrightarrow^{*}\mathtt{casterror}$, or
  \item for all $r$ such that $\cc[e]\mid\emptyset\longrightarrow^{*}r$, have $r = e'\mid\mu'$ and there exists $r'$ such that $e'\mid\mu'\longrightarrow r'$.
  \end{itemize}
\end{theorem}

This theorem states that no program can ever result in
$\mathtt{pyerror}(\circ)$, which would indicate an uncaught type error
in translated Anthill code. Type soundness in the usual sense is an
immediate corollary of this theorem, by choosing $\cc$ to be the empty
context $\chole$.

% In order to prove this theorem, several lemmas are needed. First, we
% must show that the translation relation is type-preserving.

% \begin{lemma}[Anthill translation preserves typing]\label{lem:trans-type-preserve}
%   If $\Gamma\vdash t\leadsto e:\tsym$, then $\flr{\Gamma};\emptyset\vdash
%   e:\lfloor \tsym\rfloor$.
% \end{lemma}
% We also must show that preservation holds in the usual sense.

% \begin{lemma}[Preservation]\label{lem:preservation}
%   If $\Gamma;\Sigma\vdash e:S$, $\Sigma\vdash\mu$, and $e\mid\mu\longrightarrow e'\mid\mu'$, then $\Gamma;\Sigma'\vdash e:S$ and $\Sigma'\vdash\mu'$.
% \end{lemma}

% The property that $\texttt{pyerror}(\circ)$ cannot occur from
% well-typed \upython{} code is shown by the progress
% lemma.

% \begin{lemma}[Progress]\label{lem:progress}
%   If $\emptyset;\Sigma \vdash e:S$ and $\Sigma\vdash\mu$, then either $e$ is a value, or $e\mid\mu \longrightarrow r$ and:
%   \begin{itemize}
%   \item $r=(e'\mid\mu')$ for some expression $e'$ and heap $\mu'$;
%   \item $r=\mathtt{casterror}$; or
%   \item $r=\mathtt{pyerror}(\bullet)$.
%   \end{itemize}
% \end{lemma}

% Finally, plugging a well-typed expression into a well-typed context
% results in an appropriately typed expression.

% \begin{lemma}[Composition]\label{lem:composition}
%   If $\cc: \Gamma{;}S_1\Rightarrow \emptyset{;}S_2$ and
%   $\Gamma{;}\emptyset\vdash e:S_1$, then $\emptyset{;}\emptyset\vdash \cc[e]:S_2$.
% \end{lemma}

We proved this theorem for the Anthill and \upython{} languages using
the Coq proof assistant; the completed proof files are available at \hosturl. A sketch of the
proof can also be examined in Appendix \ref{apx:proofs}. The proof
combines a typical progress and preservation type soundness proof (for
\upython, and so including many additional cases for errored terms)
with proofs that the translation relation is type preserving and that
composition of terms and contexts is well-typed.

\subsection{Ramifications of open world soundness}

Because Anthill admits open world soundness, a program written in
Anthill can be used by native \upython{} clients. For example, an
Anthill library can put type annotations on its API boundaries, and
these types will be checked, preventing difficult-to-diagnose errors
from arising deep within the library --- even if the library is used
by code which has no concept of static types. Furthermore, the Anthill
code is protected from errors arising due to mutation. While foreign
functions are not modeled directly in the Anthill and \upython{}
calculi, note that the distinction between untyped Python programs and
compiled C code is relevant in \guarded{} because of the presence of
proxies. Since Anthill and \upython{} lack proxies, however, this
distinction is irrelevant, and such foreign functions can be modeled
as untranslated \upython{} code, and thus open world soundness shows
that calls to such code within Anthill programs will not interfere
with Anthill's type soundness.

We conjecture that Reticulated Python is also open world sound when
using the \transient{} semantics; it is certainly closer to open world
soundness than the \guarded{} semantics. Evidence for this exists in
the work of \citet{vitousek:2014retic}. They found that Reticulated
Python ran all of their case studies without error under \transient{},
while it could not successfully execute one of them when using
\guarded{}. Further, more modifications to the typed programs were
necessary when using \guarded{} to avoid errors caused by proxy identity problems.

Combined with its ease of implementation, the fact that \transient{}
admits open world soundness means that it is a realistic and useful
technique in designing gradually typed languages like Reticulated
Python. It does require performance overhead in its pervasive checks,
and we envision practical implementations offering a ``switch'' to
allow developers to debug their program with thorough type checking
and disable it for production (similar to, but more complete than,
Dart's checked mode). Open world soundness demonstrates that it is
useful in circumstances that challenge \guarded{}, and its ease of
implementation makes it more practical than other approaches.

\section{Conclusions}\label{sec:conclusions}

The traditional \guarded{} approach for the runtime semantics of
gradually typed languages, based on proxies, is well-understood and
powerful, but it is unsound in an open world in when applied to
languages like Python.  The \transient{} design for gradual typing
provides an alternative approach which is open world sound. 

In this paper we develop a formal treatment of \transient{}
with calculi that model Reticulated Python and Python. Furthermore, we
discuss a formal property that is relevant to source-to-source
implementations of gradually typed languages, open world soundness,
which states that typed code can be embedded in untranslated code
without causing uncaught type errors.  We prove that our calculi admit
the open world soundness principle and mechanize the proof in Coq.

Many industrial gradually typed languages avoid runtime checking
altogether. This is especially relevant when the system is
designed to translate to an underlying language whose semantics cannot
be modified by the designer of the surface language.  Reticulated
Python and the \transient{} semantics demonstrate that gradual typing
can be implemented straightforwardly, without
modifying the target language's semantics, and while allowing
interoperation and preserving soundness.

% Cam provided a fantastic and highly detailed line edit of the paper for ECOOP
% Ambrose reviewed the paper for OOPSLA
% Cam wailed on the related work section for OOPSLA
% Matteo reviewed the paper for OOPSLA
% Andre reviewed the paper for POPL
% Chris reviewed the paper for ECOOP
%
%
% \paragraph*{Acknowledgments} We are grateful to Sam Tobin-Hochstadt,
% Cameron Swords, Ambrose Bonnaire-Sergeant, Matteo Cimini, Andre
% Kuhlenschmidt, and Chris Wailes for their feedback during the
% development of this work, and to Jim Baker for his assistance and
% feedback in designing Reticulated Python.
{
\scriptsize
\bibliographystyle{plainnat}
\bibliography{mike,all,library}
}

\clearpage
\newpage
\renewcommand\rightmark{}
\renewcommand\leftmark{}
\appendix
\section{Appendix: Full semantics} \label{apx:full-semantics}

Figure \ref{fig:apx:full-translation} shows the translation from Anthill
Python to \upython, and Figure \ref{fig:apx:upython-blame-typing} shows
the type system used to prove soundness for \upython.
% Figure \ref{fig:upython-subtyping} contains the subtyping relation for \upython{} types $S$. 
Figures \ref{fig:apx:lookup-owned} and \ref{fig:apx:eval} show the runtime
semantics of \upython, and are identical to Figures \ref{fig:lookup}
and \ref{fig:eval} in the main body of the paper except for the
addition of origin labels $p$. Figure \ref{fig:apx:context-typing}
shows the typing judgments for contexts $\cc$. Figure
\ref{fig:apx:more-metafunctions} shows assorted metafunctions not
defined in the main body of this work.

\begin{figure*}
\boxed{\cig\vdash t\leadsto e:\tsym}
\begin{mathpar}
  \ninference{IVar}{\Gamma(x)=\tsym}{\cig\vdash x\leadsto x:\tsym}\and
\ninference{IInt}{ }{\cig\vdash n \leadsto n:\mathtt{int}}\and
\ninference{ILet}{\cig\vdash t_1\leadsto e_1:\tsym_1 \\ \Gamma,x{:}\tsym_1\vdash t_2\leadsto e_2:\tsym_2}
          {\cig\vdash \lett{x}{t_1}{t_2}\leadsto \lett{x}{e_1}{e_2}:\tsym_2} 
\and
\ninference{IGet}{\cig\vdash t\leadsto e:\tsym_1 \\ \Delta=\textit{mems}(\tsym_1) \\\\ \Delta(x)=\tsym_2}
          {\cig\vdash t.x \leadsto \tcheck{e.x^\circ}{\flr{\tsym_2}}:\tsym_2}
\and
\ninference{IGet-Check}{\cig\vdash t\leadsto e: \tsym_1 \\ \Delta=\textit{mems}(\tsym_1)\\\\ x\not\in\textit{dom}(\tsym_1) \\ \textit{queryable}(\tsym_1)=\lozenge}
          {\cig\vdash t.x \leadsto (\tcheck{e}{\bobjty{x}}).x^\circ:\dyn}
\and
\ninference{ISet}{\cig\vdash t_1\leadsto e_1:\tsym_1 \\ \cig\vdash t_2 \leadsto e_2:\tsym_2' \\\\ \Delta=\textit{mems}(\tsym_1)}
          {\cig\vdash t_1.x=t_2\leadsto e_1.x^\circ=e_2:\mathsf{int}}
\and
\ninference{ISet-Check}{\cig\vdash t_1 \leadsto e_1:\tsym_1 \\ \cig\vdash t_2\leadsto e_2:\tsym_2 \\\\ \Delta=\textit{mems}(\tsym_1) \\ x\not\in\textit{dom}(\tsym_1) \\\\ \textit{queryable}(\tsym_1)=\lozenge}
          {\cig\vdash t_1.x=t_2\leadsto (\tcheck{e_1}{\bobjty{\emptyset}}).x^\circ=e_2:\mathsf{int}}
\and
\fninference{IFun}{\Gamma,\overline{x:\tsym_1} \vdash t\leadsto e:\tsym_2' \\ \tsym_2' \lesssim \tsym_2}
            {\Gamma {\vdash} \funt{\overline{x{:}\tsym_1}}{\tsym_2}{t} {\leadsto} \fune{\overline{x}}{\overline{\mathtt{let}~x=\tcheck{x}{\flr{\tsym_1}}~\mathtt{in}~}e}{:}\overline{\tsym_1}{\to} \tsym_2}
\and
\ninference{IApp-Dyn}{\cig\vdash t_1\leadsto e_1:\dyn \\ \cig\vdash \overline{t_2\leadsto e_2:\tsym}}
          {\cig\vdash t_1(\overline{t_2}) \leadsto (\tcheck{e_1}{\bto{|\overline{\tsym}|}})(\overline{e_2})^\circ : \dyn}
\and
\ninference{IApp-Fun}{\cig\vdash t_1\leadsto e_1:\overline{\tsym_1} \to \tsym_2 \\ 
           \Gamma\vdash \overline{t_2\leadsto e_2:\tsym_1'} \\\\ |\overline{\tsym_1}| = |\overline{t_2}| \\ 
           \overline{\tsym_1' \lesssim \tsym_1}}
          {\cig\vdash t_1(\overline{t_2})\leadsto \tcheck{(e_1(\overline{e_2})^\circ)}{\flr{\tsym_2}} : \tsym_2}
\and
\ninference{IApp-Constr}{\cig\vdash t_1\leadsto e_1:\classty{X}{q}{\Delta_1}{\Delta_2}{\overline{\tsym_1}} \\
           \Gamma\vdash \overline{t_2\leadsto e_2:\tsym_1'} \\\\ |\overline{\tsym_1}| = |\overline{t_2}| \\
           \overline{\tsym_1' \lesssim \tsym_1} \\ \tsym_2 = \objty{X}{q}{\textit{instantiate}(\Delta_1,\Delta_2)}}
          {\cig \vdash t_1(\overline{t_2}) \leadsto \tcheck{(e_1(\overline{e_2})^\circ)}{\flr{\tsym_2}}:\tsym_2}
\and
  \fninference{IClass}{\cig\vdash \overline{t_s\leadsto e_s:\tsym_s} \\
           \tsym_{class}=\classty{X}{q}{\Delta_1}{\Delta_2}{\overline{\tsym_c}}\\
\cig\vdash_\sigma c \leadsto e_c : \overline{\tsym_c} \\ \Gamma; \tsym_{class}\vdash_\varsigma \overline{m\leadsto e_m:\tsym_m} \\ \cig\vdash \overline{t_f\leadsto e_f:\tsym_f} \\
           \overline{e_s'=\tcheck{e_s}{\bclassty{\flr{\textit{mems}(\tsym_s)}}{\mathsf{Any}}}} \\ 
           \forall x \in \mathit{dom}(\Delta_1),~ (\overline{\ell_f}\times \overline{A_f} \cup 
           \overline{\ell_m}\times \overline{A_m} \cup \overline{\mathit{mems}(A_s)})(x) \lesssim \Delta_1(x)}
          {
              \cig\vdash \classt{X}{q}{\Delta_1}{\Delta_2}{\overline{t_s}}{\overline{\ell_f=^M m}}{\overline{\ell_m=^F t_m}}{c} 
              \leadsto \\\\ \classe{X}{\overline{e_s'}}{\overline{\ell_m=e_m},\overline{\ell_f=e_f}}{e_c}{}:\tsym_{class}
            }
\end{mathpar}
\boxed{\cig\vdash_\sigma c\leadsto e:\overline{\tsym}}\vspace{-3ex}
\begin{mathpar}
\fninference{IConstruct}{\Gamma,x_s{:}\dyn,\overline{x{:}\tsym_1}\vdash t\leadsto e:\tsym_2}
          {\cig\vdash_\sigma \constructt{\overline{x{:}\tsym_1}}{t} \leadsto \fune{x_s,\overline{x}}{\overline{\mathtt{let}~x=\tcheck{x}{\flr{\tsym_1}}~\mathtt{in}~}e}:\overline{\tsym_1}}
\end{mathpar}
\boxed{\Gamma;\tsym\vdash_\varsigma m\leadsto e:\tsym}\vspace{-1ex}
\begin{mathpar}
\ninference{IMethod}{\tsym_o = \objty{X}{q}{\textit{instantiate}(\Delta_1,\Delta_2)} \\ \tsym_2' \lesssim \tsym_2\\ \Gamma,x_s{:}\tsym_o,\overline{x{:}\tsym_1}\vdash t\leadsto e:\tsym_2'}
          {
              \Gamma; \classty{X}{q_1}{\Delta_1}{\Delta_2}{\overline{\tsym_c}}\vdash_\varsigma \method{\overline{x{:}\tsym_1}}{\tsym_2}{t} \leadsto \\\\
              \fune{x_s,\overline{x}}{\lett{x_s}{\tcheck{x_s}{\flr{\tsym_o}}}{\overline{\mathtt{let}~x=\tcheck{x}{\flr{\tsym_1}}~\mathtt{in}~}e}:\overline{\tsym_1}\to \tsym_2}
            }
\end{mathpar}
  \caption{Translation from Anthill Python to \upython{} (including origin)}
  \label{fig:apx:full-translation}
\end{figure*}

\begin{figure*}
\boxed{\Gamma;\Sigma\vdash e:S}
\begin{mathpar}
\ninference{TSubsump}{\tlg\vdash e:S_2 \\ S_2 <: S_1}{\tlg\vdash e:S_1}\and
\ninference{TVar}{\Gamma(x)=S}{\tlg\vdash x{:}S}\and
\ninference{TAddr}{\Sigma(a)=S}{\tlg\vdash a{:}S}\and
\ninference{TInt}{ }{\tlg\vdash n{:}\mathsf{int}}\and
\ninference{TApp-Dyn}{\tlg\vdash e_1:\bdyn \\ \tlg\vdash \overline{e_2:\bdyn}}{\tlg\vdash e_1(\overline{e_2})^\bullet : \bdyn}\and
\ninference{TApp}{\tlg\vdash e_1:\bto{n} \\ \tlg\vdash \overline{e_2:\bdyn} \\ |\overline{e_2}| =n}{\tlg\vdash e_1(\overline{e_2})^\circ : \bdyn}
\and
\ninference{TGet-Dyn}{\tlg\vdash e:\bdyn}
          {\tlg\vdash e.x^\bullet:\bdyn}
          \and
\ninference{TGet}{\tlg\vdash e:\bobjty{x}}
          {\tlg\vdash e.x^\circ:\bdyn}
\and
\ninference{TSet-Dyn}{\tlg\vdash e_1:\bdyn \\ \tlg\vdash e_2:\bdyn}
          {\tlg\vdash e_1.x^\bullet=e_2:\mathsf{int}}
\and
\ninference{TSet}{\tlg\vdash e_1:\bobjty{\emptyset} \\ \tlg\vdash e_2:\bdyn}
          {\tlg\vdash e_1.x^\circ=e_2:\mathsf{int}}
\and
\ninference{TClass-Dyn}{\tlg\vdash \overline{e_s:\bdyn} \\ \tlg\vdash \overline{e_m:\bdyn} \\ \tlg\vdash e_c:\bdyn}
          {\tlg\vdash \classe{X}{\overline{e_s}}{\overline{x=e_m}}{e_c}{\bullet}:\bclassty{\overline{x}}{\mathsf{Any}}}
\and
\ninference{TClass}{\tlg\vdash \overline{e_s:\bclassty{\delta}{\mathsf{Any}}} \\ \delta' = \bigcup\overline{\delta} \\ \tlg\vdash \overline{e_m:\bdyn} \\ \tlg\vdash e_c:\bto{(n{+}1)}}
          {\tlg\vdash \classe{X}{\overline{e_s}}{\overline{x=e_m}}{e_c}{\circ}:\bclassty{\overline{x} \cup \delta'}{n}}
\and
\ninference{TFun}{\Gamma,\overline{x{:}\bdyn};\Sigma \vdash e:\bdyn}{\tlg \vdash \fune{\overline{x}}{e} :\bto{|\overline{x}|}}
\and
\ninference{TCheck}{\tlg\vdash e:\bdyn}{\tlg\vdash \tcheck{e}{S} : S}\and
\ninference{TLet}{\tlg\vdash e_1:S_1 \\ \Gamma,x{:}S_1{;}\Sigma\vdash e_2:S_2}{\tlg\vdash \lett{x}{e_1}{e_2} : S_2}
\end{mathpar}
\boxed{\Sigma;\mu\vdash a:S}
\begin{mathpar}
\ninference{THClass}{\mu(a)=\classh{\overline{a'}}{\overline{x_f=v_f}}{v} \\ \mathit{hasattrs}(a,\delta,\mu) \\ \textit{param-match}(a,\mu,C) \\\overline{\Sigma(a')=\bclassty{\delta'}{C'}} \\ \overline{\emptyset{;}\Sigma\vdash v_f:\bdyn}}
                     {\Sigma;\mu\vdash a:\bclassty{\delta}{C}}\and
\ninference{THObject}{\mu(a)=\objh{a'}{\overline{x_f=v_f}} \\ \mathit{hasattrs}(a,\delta,\mu)\\ \Sigma(a')=\bclassty{\delta'}{C'} \\ \overline{\emptyset{;}\Sigma\vdash v_f:\bdyn}}
                     {\Sigma;\mu\vdash a:\bobjty{\delta}}
% \ninference{THClass}{\Sigma\vdash\mu \\ \emptyset{\mid}\Sigma\vdash v:\bto{n} \\ \overline{\Sigma(a)=\bclassty{\delta}{\mathsf{Any}}} \\ \delta'=\bigcup\overline{\delta}}
%                      {\Sigma,a:\bclassty{\textit{dom}(M) \cup \delta'}{n}\vdash \mu[a\mapsto \classh{\overline{a}}{M}{v}]}\and
% \ninference{TEObject}{\Sigma\vdash\mu \\ \emptyset{\mid}\Sigma\vdash v:\bto{n} \\ \Sigma(a)=\bclassty{\delta}{\mathsf{Any}}}
%                       {\Sigma,a:\bobjty{\textit{dom}(M) \cup \delta}\vdash \mu[a\mapsto \objh{a}{M}]}\and
% \ninference{TEEmpty}{}{\emptyset\vdash \cdot}
\end{mathpar}
\boxed{\Sigma\vdash \mu}
\begin{mathpar}
\inference{\textit{dom}(\Sigma)=\textit{dom}(\mu) \\ \forall a \in \textit{dom}(\Sigma).~\Sigma;\mu\vdash a:\Sigma(a)}{\Sigma\vdash \mu}
\end{mathpar}
  \caption{Type system for \upython{} (including origin)}
  \label{fig:apx:upython-blame-typing}
\end{figure*}

% \begin{figure*}
% \boxed{S <: S}
%   \begin{mathpar}
%     \inference{ }{S <: \bdyn} \and \inference{ }{S <: S} \and \inference{S_1 <: S_2 \\ S_2 <: S_3}{S_1 <: S_3}\and
%     \inference{\delta_2 \subseteq \delta_1}{\bobjty{\delta_1} <: \bobjty{\delta_2}}\and
%     \inference{ }{\bclassty{\delta}{n} <: \bclassty{\delta}{\mathsf{Any}}}\and
%     \inference{\delta_2 \subseteq \delta_1}{\bclassty{\delta_1}{C} <: \bclassty{\delta_2}{C}}\and
%     \inference{ }{\bclassty{\delta}{C} <: \bobjty{\delta}}\and
%     \inference{ }{\bclassty{\delta}{n} <: \bto{n}}
%   \end{mathpar}
%   \caption{Subtyping in \upython}
%   \label{fig:upython-subtyping}
% \end{figure*}
\newcommand{\loca}[3]{\left\{\begin{array}{p{4.2cm}ll}\ensuremath{#1} & #2 & #3\end{array}\right.}

\begin{figure}
\boxed{\mathit{lookup}(a,h,\ell,\mu,p)=r}
  \begin{mathpar}
    \inferrule{M(\ell)=v}{\mathit{lookup}(a,\objh{a'}{M},\ell,\mu,p)=v\mid\mu}\and
    \inferrule{\ell\not\in\mathit{dom}(M) \\\\ \mathit{getattr}(a,\ell,\mu)=v \\ v\neq\lambda\overline{\ell}.e}{\mathit{lookup}(a,\objh{a'}{M},\ell,\mu,p)=v\mid\mu}\and
    \inferrule{\ell\not\in\mathit{dom}(M)\\\\\mathit{getattr}(a,\ell,\mu)=v \\ v=\lambda\overline{x}.e \\ |\overline{y}|=|\overline{x}|-1}{\mathit{lookup}(a,\objh{a'}{M},\ell,\mu,p)=\lambda\overline{y}.v(a,\overline{y})^p\mid\mu}\and
    \inferrule{\ell\not\in\mathit{dom}(M) \\ \mathit{getattr}(a,\ell,\mu)=v \\ v=\lambda\overline{x}.e \\ |\overline{x}| =0}{\mathit{lookup}(a,\objh{a'}{M},\ell,\mu,p)=\mathtt{casterror}}\and
    \inferrule{\textit{getattr}(a,\ell,\mu)=v}{\mathit{lookup}(a,\classh{\overline{a'}}{M}{v'},\ell,\mu,p)=v\mid\mu}
  \end{mathpar}
\caption{Relations used in \upython{} evaluation, with origin labels}
\label{fig:apx:lookup-owned}
\end{figure}

\begin{figure}
\boxed{e\mid\mu\longrightarrow r}
  \begin{mathpar}
    \ninference{EStep}{e\mid\mu\longrightarrow e'\mid\mu'}
              {E[e]\mid\mu \longmapsto E[e']\mid\mu'}
\and
    \ninference{EPyError}{e\mid\mu\longrightarrow \mathtt{pyerror}(p)}
              {E[e]\mid\mu \longmapsto \mathtt{pyerror}(p)}
\and
    \ninference{ECastError}{e\mid\mu\longrightarrow \mathtt{casterror}}
              {E[e]\mid\mu \longmapsto \mathtt{casterror}}
  \end{mathpar}
\newcommand{\steplab}[1]{\scriptsize\textsc{(#1)}}
\[
\begin{array}{lcl}
%  \multicolumn{4}{l}{\textsc{Checks}}\\
   \steplab{ECheck1,2}\hfill\tcheck{v}{S}\mid\mu \hfill& \longrightarrow & \left\{
      \begin{array}{p{2.2cm}l}
        $v \mid \mu$ & \text{if }\mathit{check}(v,\mu,S)\\[.5ex]
        $\mathtt{casterror}$ & \text{otherwise}
      \end{array}\right.\\
 \steplab{EApp1,2,3}\hfill v_1(\overline{v_2})^p\mid\mu & \longrightarrow & \left\{
    \begin{array}{p{2.2cm}l}
      $e[\overline{x/v_2}] \mid \mu$ & \text{if }v_1=\lambda \overline{x}.e\\
      &\text{and } |\overline{x}| =|\overline{v_2}|\\[.5ex]
      $\mathtt{let}~\_=$ & \text{if } v_1=a \\
      \and$v'(a',\overline{v_2})^p$ & \text{and } \mu(a)=\classh{\overline{a''}}{M}{v'}\\
      \and$\mathtt{in}~a' \mid \mu'$& \text{and } a' \text{ fresh}\\
      & \text{and } \mu'=\mu[a'\mapsto \objh{a}{\emptyset}]\\[.5ex]
      $\mathtt{pyerror}(p)$ & \text{otherwise}
    \end{array}\right.\\[.5ex]
  \steplab{ELet}\hfill\lett{x}{v}{e} & \longrightarrow & \and\;e[x/v]\\[.5ex]
  \overset{\steplab{EClass1,2}\hfill}{\classe{X}{\overline{a}}{M}{v}{p} \mid \mu} & \longrightarrow &  \left\{
    \begin{array}{p{2.2cm}l}
      $a'\mid\mu[a'\mapsto h]$ & \text{if } \overline{\mu(a)=\classh{\overline{a''}}{M'}{v'}} \\
      & \text{and } \textit{param-match}(v,\mu,\textsf{Any})\\
      & \text{and } h=\classh{\overline{a}}{M}{v}\\
      & \text{and } a' \text{ fresh}\\[.5ex]
      $\mathtt{pyerror}(p)$ & \text{otherwise}
    \end{array}\right.\\
  \overset{\steplab{EClass3}\hfill}{\classe{X}{\overline{v_1}}{M}{v_2}{p} \mid \mu} & \longrightarrow &
  \begin{array}{ll}
    \;\;\;\mathtt{pyerror}(p) & \and\;\text{if } \overline{v_1 \neq a}
  \end{array}\\[.5ex]
  \steplab{EGet1,2}\hfill a.\ell^p\mid\mu & \longrightarrow & \left\{
    \begin{array}{p{2.2cm}l}
      $r$ & \text{if } \mathit{lookup}(a,\mu(a),\ell,\mu)=r\\
      $\mathtt{pyerror}(p)$ & \text{otherwise}
    \end{array}\right.\\
  \steplab{EGet3}\hfill v.\ell^p\mid\mu & \longrightarrow &
  \begin{array}{ll}
    \;\;\;\mathtt{pyerror}(p) & \and\;\text{if } v\neq a
  \end{array}\\
   \steplab{ESet1,2,3}\hfill a.\ell^p=v\mid\mu & \longrightarrow & \left\{
    \begin{array}{p{2.2cm}l}
      $0\mid\mu[a\mapsto h']$ & \text{if } \mu(a)=\objh{a'}{M} \\
      & \text{and } h'=\objh{a'}{M[\ell=v]}\\[.5ex]
      $0 \mid \mu[a\mapsto h']$ & \text{if } \mu(a)=\classh{\overline{a'}}{M}{v'}\\
      & \text{and } h'=\classh{\overline{a'}}{M[\ell=v]}{v'}\\[.5ex]
      $\mathtt{pyerror}(p)$ & \text{otherwise}
    \end{array}\right.\\
  \steplab{ESet4}\hfill v_1.\ell^p=v_2\mid\mu & \longrightarrow &
  \begin{array}{ll}
    \;\;\;\mathtt{pyerror}(p) & \and\;\text{if } v_1\neq a
  \end{array}\\
\end{array}
\]
\boxed{e\mid\mu\longrightarrow^{*}r}
\begin{mathpar}
  \fninference{MRefl}{ }{e\mid\mu\longrightarrow^{*}e\mid\mu}\and
  \fninference{MPyErr}{e\mid\mu\longmapsto\mathtt{pyerror}(p)}{e\mid\mu\longrightarrow^{*}\mathtt{pyerror}(p)}\and
  \fninference{MCastErr}{e\mid\mu\longmapsto\mathtt{casterror}}{e\mid\mu\longrightarrow^{*}\mathtt{casterror}}\and
  \fninference{MChain}{e\mid\mu\longmapsto e'\mid\mu' \\ e'\mid\mu'\longrightarrow^{*}r}{e\mid\mu\longrightarrow^{*}r}
\end{mathpar}
  \caption{\upython{} evaluation rules, with origin labels}
  \label{fig:apx:eval}
\end{figure}

\newcommand{\elg}{\Gamma'{;}\emptyset}
\newcommand{\erinference}[3]{\inference{#2}{#3}}

\begin{figure*}
\boxed{\cc:\Gamma{;}S\Rightarrow\Gamma{;}S}
  \begin{mathpar}
    \inference{ }{\chole:\Gamma{;}S\Rightarrow \Gamma{;}S}\and
\erinference{TSubsump}{\cc:\Gamma{;}S_1\Rightarrow \Gamma'{;}S_3 \\ S_3 <: S_2}{\cc:\Gamma{;}S_1\Rightarrow \Gamma'{;}S_2}\and
\erinference{TApp-Dyn}{\cc:\Gamma{;}S\Rightarrow\Gamma'{;}\bdyn \\ \Gamma'{;}\emptyset\vdash\overline{e:\bdyn}}{\cc(\overline{e})^\bullet : \Gamma{;}S\Rightarrow\Gamma'{;}\bdyn}\\
\erinference{TApp-Dyn}{\Gamma'{;}\emptyset\vdash e:\bdyn \\ \Gamma'{;}\emptyset\vdash \overline{e:\bdyn} \\ \cc:\Gamma{;}S\Rightarrow\Gamma'{;}\bdyn}{e(\overline{e},\cc,\overline{e})^\bullet : \Gamma{;}S\Rightarrow\Gamma'{;}\bdyn}
\and
\inference{\cc:\Gamma{;}S\Rightarrow\Gamma'{;}\bdyn}
          {\cc.x^\bullet:\Gamma{;}S\Rightarrow\Gamma'{;}\bdyn}
\and
\inference{\cc:\Gamma{;}S\Rightarrow\Gamma'; \bdyn \\ \elg\vdash e:\bdyn}
          {\cc.x^\bullet=e:\Gamma{;}S\Rightarrow \Gamma';\mathsf{int}}\and
\inference{\elg\vdash e:\bdyn \\ e:\Gamma{;}S\Rightarrow\Gamma'{;}\bdyn}
          {\tlg\vdash e.x^\bullet=\cc:\Gamma{;}S\Rightarrow\Gamma';\mathsf{int}}
\and
\inference{\elg\vdash \overline{e_s:\bdyn} \\ \elg\vdash \overline{e_m:\bdyn} \\ \elg\vdash e_c:\bdyn \\ \cc:\Gamma{;}S\Rightarrow\Gamma'{;}\bdyn}
          {\classe{X}{\overline{e_s},\cc,\overline{e_s}}{\overline{x=e_m}}{e_c}{\bullet}:\Gamma{;}S\Rightarrow\Gamma'{;}\bdyn}
\and
\inference{\elg\vdash \overline{e_s:\bdyn} \\ \elg\vdash \overline{e_m:\bdyn} \\ \elg\vdash e_c:\bdyn \\ \cc:\Gamma{;}S\Rightarrow\Gamma'{;}\bdyn}
          {\classe{X}{\overline{e_s}}{\overline{x=e_m},x=\cc,\overline{x=e_m}}{e_c}{\bullet}:\Gamma{;}S\Rightarrow\Gamma'{;}\bdyn}
\and
\inference{\elg\vdash \overline{e_s:\bdyn} \\ \elg\vdash \overline{e_m:\bdyn} \\ \cc:\Gamma{;}S\Rightarrow\Gamma'{;}\bdyn}
          {\classe{X}{\overline{e_s}}{\overline{x=e_m}}{\cc}{\bullet}:\Gamma{;}S\Rightarrow\Gamma'{;}\bdyn}
\and
\inference{\cc:\Gamma{;}S\Rightarrow \Gamma',\overline{x{:}\bdyn};\bdyn}{\fune{\overline{x}}{\cc}:\Gamma{;}S\Rightarrow\Gamma';\bto{|\overline{x}|}}
\and
\erinference{TCheck}{\cc:\Gamma{;}S\Rightarrow\Gamma'{;}\bdyn}{\tcheck{\cc}{S'} : \Gamma{;}S\Rightarrow\Gamma'{;}S'}\\
\erinference{TLet}{\cc:\Gamma{;}S_1\Rightarrow\Gamma'{;}S_2 \\ \Gamma',x{:}S_2{;}\emptyset\vdash e:S_3}{\lett{x}{\cc}{e} : \Gamma{;}S_1\Rightarrow \Gamma'{;}S_3}\and
\erinference{TLet}{\Gamma'{;}\emptyset\vdash e:S_2 \\ \cc:\Gamma{;}S_1\Rightarrow\Gamma',x{:}S_2; S_3}{\lett{x}{e}{\cc} : \Gamma{;}S_1\Rightarrow \Gamma'{;}S_3}
  \end{mathpar}
  \caption{Type system for \upython{} contexts}
  \label{fig:apx:context-typing}
\end{figure*}

\section{Appendix: Open World Soundness}\label{apx:proofs}

The full mechanized proof of open-world soundness for Anthill and
\upython{} is presented in the Coq files \texttt{anthill1.v},
\texttt{anthill2.v}, \texttt{anthill3.v}, and \texttt{anthill4.v}. The
semantics of the mechanized proof differ from those presented here in
that the reduction is single-step without evaluation contexts, and the
receiver is the last parameter of constructors and methods, rather
than the first.  The Coq proof also does not use a heap type $\Sigma$,
but instead uses the heap directly in the \upython{} typechecking
relation and the theorems (but ignoring the values of object fields).

This section contains proof sketches of the key lemmas and most
important technical lemmas, as well as the proof of open-world
soundness.

\setcounter{theorem}{0}
\begin{techlemma}[Environment weakening]\label{lem:apx:weakening}
If $\Gamma{;}\Sigma\vdash e:S$ and $\Gamma \subseteq  \Gamma'$ then $\Gamma';\Sigma\vdash e:S$.
\end{techlemma}
\begin{sketch}
  Straightforward induction on $\Gamma{;}\Sigma\vdash e:S$.
\end{sketch}

\begin{mainlemma}[Anthill translation preserves typing]\label{lem:apx:trans-type-preserve}\hfill
  \begin{itemize}
  \item If $\Gamma\vdash t\leadsto e:\tsym$, then $\flr{\Gamma};\emptyset\vdash
  e:\lfloor \tsym\rfloor$.
\item If $\Gamma\vdash c\leadsto e: \overline{\tsym}$, then
  $\flr{\Gamma};\emptyset\vdash e:\bto{|\overline{\tsym}|}$.
\item If $\Gamma{;} \tsym_1 \vdash d \leadsto e:\tsym_2$, then
  $\flr{\Gamma};\emptyset\vdash e:\flr{\tsym_2}$
  \end{itemize}
\end{mainlemma}
\begin{sketch}
  By induction on $\Gamma\vdash t\leadsto e:\tsym$, $\Gamma\vdash
  c\leadsto e:\overline{\tsym}$, and $\Gamma{;} \tsym_1\vdash d\leadsto
  e:\tsym_2$. In the \textsc{IFun}, \textsc{IMethod}, and
  \textsc{IConstruct} cases, use Technical Lemma \ref{lem:apx:weakening}.
\end{sketch}

\begin{techlemma}\label{lem:apx:lemmaf}
  If $\Sigma\vdash\mu$ and $\emptyset;\Sigma\vdash a:S$, then $a\in\textit{dom}(\Sigma)$.
\end{techlemma}
\begin{sketch}
  By induction on $\emptyset;\Sigma\vdash a:S$.
\end{sketch}

\begin{techlemma}[Canonical forms]\label{lem:apx:canonical}
  If $\emptyset{;}\Sigma\vdash v : S$ and $\Sigma\vdash\mu$, then:
  \begin{itemize}
  \item If $S=\mathsf{int}$, then $v=n$.
  \item If $S=\bto{n}$, then either 
    \begin{itemize}
    \item $v=\lambda \overline{x}.e$ and $|\overline{x}| =n$, or
    \item $v=a$ and $\mu(a)=\classh{\overline{a'}}{M}{v'}$ and
      $\textit{param-match}(v', \mu, n)$ 
    \end{itemize}
  \item If $S=\bclassty{\delta}{C}$, then $v=a$ and \\$\mu(a)=\classh{\overline{a'}}{M}{v'}$ and
      \\$\textit{param-match}(v', \mu, C)$ and $\textit{hasattrs}(a,\delta,\mu)$.
  \item If $S=\bobjty{\delta}$ then $v=a$ and either
    \begin{itemize}
    \item $\mu(a)=\classh{\overline{a'}}{M}{v'}$ and
      \\$\textit{hasattrs}(a,\delta,\mu)$, or
    \item $\mu(a)=\objh{a}{M}$ and
      $\textit{hasattrs}(a,\delta,\mu)$.
    \end{itemize}
  \item If $S=\bdyn$ then for some $S'\neq\bdyn$, $S' <: S$,
    $\emptyset{;}\Sigma\vdash v : S'$
  \end{itemize}
\end{techlemma}
\begin{sketch}
  By induction on $\emptyset;\Sigma\vdash v : S$. In case
  \textsc{TSubsump}, proceeding by cases on $S$ and frequently using
  Technical Lemma \ref{lem:apx:lemmaf}.
\end{sketch}

\begin{techlemma}[Heap weakening]\label{lem:apx:heap-weakening}
If $\Gamma{;}\Sigma\vdash e:S$ and $\Sigma \sqsubseteq  \Sigma'$ then $\Gamma{;}\Sigma'\vdash e:S$.
\end{techlemma}
\begin{sketch}
  By induction on $\Gamma{;}\Sigma \vdash e:S$,
  subsuming class types to object types when necessary.
\end{sketch}

\begin{techlemma}\label{lem:apx:lemmai}
  If $\Sigma\vdash\mu$ and $\Sigma;\mu\vdash a:S$ and $\textit{hasattrs}(a, \delta', \mu)$, then
  \begin{itemize}
  \item If $S=\bclassty{\delta}{C}$, then $\delta' \subseteq \delta$.
  \item If $S=\bobjty{\delta}$, then $\delta' \subseteq \delta$.
  \end{itemize}
\end{techlemma}
\begin{sketch}
  By induction on $\Sigma;\mu\vdash a:S$.
\end{sketch}

\begin{techlemma}\label{lem:apx:parammatch}
  If $\Sigma\vdash\mu$ and $\emptyset;\Sigma\vdash v:S$ and
  $\textit{param-match}(v,\mu,n)$, then
  $\emptyset;\Sigma\vdash v:\bto{n}$
\end{techlemma}
\begin{sketch}
  By cases on $v$.
\end{sketch}

\begin{techlemma}[Substitution]\label{lem:apx:substitution}
  If $\Gamma,x:S_1{;}\Sigma\vdash e:S_2$ and
  $\emptyset{;}\Sigma\vdash v:S_1$ then $\Gamma{;}\Sigma\vdash e[x/v]:S_2$.
\end{techlemma}
\begin{sketch}
  Straightforward induction on $\Gamma,x:S_1{;}\Sigma\vdash e:S_2$.
\end{sketch}

\begin{mainlemma}[Preservation]\label{lem:apx:preservation}
  If $\emptyset{;}\Sigma \vdash e:S$, $\Sigma\vdash\mu$, and
  $e\mid\mu \longrightarrow e'\mid\mu'$, then
  $\emptyset{;}\Sigma'\vdash e':S$ and $\Sigma'\vdash\mu'$ and
  $\Sigma \sqsubseteq \Sigma'$.
\end{mainlemma}
\begin{sketch} 
  By induction on $e\mid\mu\longrightarrow e'\mid\mu'$. Uses Technical
  Lemmas \ref{lem:apx:lemmaf}, \ref{lem:apx:canonical},
  \ref{lem:apx:heap-weakening}, \ref{lem:apx:lemmai},
  \ref{lem:apx:parammatch}, and \ref{lem:apx:substitution}.
\end{sketch}

\begin{corollary}[Iterated preservation]\label{lem:apx:iterpres}
  If $\emptyset{;}\Sigma \vdash e:S$, $\Sigma\vdash\mu$, and
  $e\mid\mu \longrightarrow^{*} e'\mid\mu'$, then
  $\emptyset{;}\Sigma'\vdash e':S$ and $\Sigma'\vdash\mu'$ and
  $\Sigma \sqsubseteq \Sigma'$.
\end{corollary}
\begin{sketch}
  Straightforward induction on $e\mid\mu \longrightarrow
  e'\mid\mu'$. In case \textsc{MChain}, apply Main Lemma \ref{lem:apx:preservation}. 
\end{sketch}

\begin{mainlemma}[Progress]\label{lem:apx:progress}
  If $\emptyset{;}\Sigma \vdash e:S$ and $\Sigma\vdash\mu$, then either $e$ is a value or $e\mid\mu\longrightarrow r$ and either:
  \begin{itemize}
  \item $r=\mathtt{pyerror}(\bullet)$, or
  \item $r=\mathtt{casterror}$, or
  \item $r=e'\mid\mu'$.
  \end{itemize}
\end{mainlemma}
\begin{sketch}
  Induction on $\emptyset{;}\Sigma \vdash e:S$. In the
  \textsc{-Dyn} cases, this lemma is trivial, since the catch-all
  evaluation rules that lead to $\texttt{pyerror}(\bullet)$ satisfy
  the lemma.
\end{sketch}

\begin{mainlemma}[Composition]\label{lem:apx:composition}
  If $\cc:\Gamma{;}S\Rightarrow \Gamma'{;}S'$ and
  $\Gamma{;}\emptyset\vdash e:S$, then $\Gamma'{;}\emptyset\vdash \cc[e]:S'$.
\end{mainlemma}
\begin{sketch}
  By induction on $\cc:\Gamma{;}S\Rightarrow \Gamma'{;}S'$.
\end{sketch}

\begin{techlemma}[Typed code blames $\bullet$]\label{lem:apx:blame}
  If $e\mid\mu\longrightarrow \mathtt{pyerror}(p)$ and $\emptyset;\Sigma\vdash e :S $ and $\Sigma\vdash\mu$, then $p=\bullet$.
\end{techlemma}
\begin{sketch}
  Straightforward induction on $e\mid\mu\longrightarrow
  \mathtt{pyerror}(p)$, noting that $\circ$-labeled code cannot take
  steps that lead to \texttt{pyerror}.
\end{sketch}
\begin{corollary}\label{lem:apx:iterblame}
  If $e\mid\mu\longrightarrow^{*} \mathtt{pyerror}(p)$ and $\emptyset;\Sigma\vdash e :S $ and $\Sigma\vdash\mu$, then $p=\bullet$.
\end{corollary}
\begin{sketch}
  Induction on $e\mid\mu\longrightarrow^{*} \mathtt{pyerror}(p)$, applying Technical Lemma \ref{lem:apx:blame} in the case for \textsc{EChain}.
\end{sketch}

\setcounter{theorem}{0}
\begin{theorem}[Open world soundness]
  If $\Gamma \vdash t \leadsto e : \tsym$, then for any context $\cc$
  such that $\cc:\flr{\Gamma}{;}\flr{\tsym}\Rightarrow \emptyset{;}S$,
  then $\emptyset;\emptyset\vdash \cc[e]:A$ and either:
  \begin{itemize}
  \item for some $v,\Sigma,\mu$, $\cc[e]\mid\emptyset\longrightarrow^{*} v\mid\mu$ and $\emptyset;\Sigma\vdash v:S$ and $\Sigma\vdash\mu$, or
  \item $\cc[e]\mid\emptyset\longrightarrow^{*}\mathtt{pyerror}(\bullet)$, or
  \item $\cc[e]\mid\emptyset\longrightarrow^{*}\mathtt{casterror}$, or
  \item for all $r$ such that $\cc[e]\mid\emptyset\longrightarrow^{*}r$, have $r = e'\mid\mu'$ and there exists $r'$ such that $e'\mid\mu'\longrightarrow r'$.
  \end{itemize}
\end{theorem}
\begin{proof}
  From Main Lemma \ref{lem:apx:trans-type-preserve},
  $\flr{\Gamma}{;}\emptyset\vdash e:\flr{\tsym}$. Then from Main
  Lemma \ref{lem:apx:composition}, $\emptyset{;}\emptyset\vdash
  \cc[e]:S$. Either $\cc[e]\mid\emptyset$ diverges or it does not. If
  it does, the theorem is satisfied immediately; otherwise, we have
  some $r$, $\cc[e]\mid\emptyset\longrightarrow^{*}r$, such that
  either $r\neq e'\mid\mu'$ or $\not\exists r',
  e'\mid\mu'\longrightarrow r'$.  In the latter case, by Main Lemma
  \ref{lem:apx:progress}, either $e'$ is a value or there exists some
  $r'$, $e'\mid\mu'\longrightarrow r'$. The latter is a contradiction,
  and so $e'$ is a value. By Corollary \ref{lem:apx:iterpres},
  $\emptyset;\Sigma'\vdash e:S$ and $\Sigma\vdash\mu'$ for some $\Sigma'$.

  If $r=\mathtt{casterror}$, the theorem is satisfied. If $r=\mathtt{pyerror}(p)$, by Corollary \ref{lem:apx:iterblame}, the theorem is satisfied.
\end{proof}

\begin{figure*}
\boxed{\textit{param-match}(v, \mu, C)}
\begin{mathpar}
  \inference{}{\textit{param-match}(\lambda \overline{x}. e, \mu, \mathsf{Any})} \and
  \inference{|\overline{x}| =n}{\textit{param-match}(\lambda \overline{x}. e, \mu, n)} \and
  \inference{\mu(a)=\classh{\overline{a}}{M}{v}}{\textit{param-match}(a, \mu, \mathsf{Any})}\and
  \inference{\mu(a)=\classh{\overline{a}}{M}{v} \\ \textit{param-match}(v, \mu, n+1)}{\textit{param-match}(a, \mu, n)}
\end{mathpar}
\boxed{\textit{getattr}(a,x,\mu)=v}
\[
\begin{array}{rcl}
  \textit{getattr}(a,x,\mu) & = & \left\{
    \begin{array}{ll}
      M(x) & \text{if}~\mu(a')=\objh{a'}{M}\\
      &\text{and}~x\in\textit{dom}(M)\\[.5ex]
      \textit{getattr}(a',x,\mu) & \text{if}~\mu(a)=\objh{a'}{M}\\
      &\text{and}~x\not\in\textit{dom}(M)\\[.5ex]
      M(x) & \text{if}~\mu(a')=\classh{\overline{a'}}{M}{v}\\
      &\text{and}~x\in\textit{dom}(M)\\[.5ex]
      \textit{getattr}(a'_k,x,\mu) & \text{if}~\mu(a)=\classh{a'_0,\ldots, a'_k, \ldots, a'_n}{M}{v}\\
      &\text{and}~x\not\in\textit{dom}(M)\\
      &\text{and}~\forall i, 0\le i<k,\textit{getattr}(a'_i,x,\mu)\neq v
    \end{array}\right.
\end{array}
\]
\begin{multicols}{2}

\boxed{\textit{hasattr}(a,x,\mu)}
\[
  \inference{\textit{getattr}(a,x,\mu)=v}{\textit{hasattr}(a,x,\mu)}
\]
\columnbreak

\boxed{\textit{hasattrs}(a,\delta,\mu)}
\[
  \inference{\forall x \in \delta.~\textit{getattr}(a,x,\mu)}{\textit{hasattrs}(a,\delta,\mu)}
\]
\end{multicols}

\boxed{\Sigma\sqsubseteq\Sigma}
\begin{mathpar}
\inference{\forall a\in\textit{dom}(\Sigma_1).~\Sigma_2(a) <: \Sigma_1(a)}{\Sigma_1\sqsubseteq\Sigma_2}
\end{mathpar}

  \boxed{\tsym\sim \tsym}
  \begin{mathpar}
    \dyn\sim \tsym \and
% I added the following. -Jeremy
    \tsym \sim \dyn \and
    \texttt{int} \sim \texttt{int} \and
% The following should be proved as a lemma! -Jeremy
%    \tsym \sim \tsym \quad
    \inference{\Delta_1\sim\Delta_2}
              {\objty{X}{q_1}{\Delta_1}\sim \objty{Y}{q_2}{\Delta_2}}\and
    \inference
        {\Delta_1\sim\Delta_3 \\ \Delta_2\sim\Delta_4 \\ 
          |\tsym_1| = |\tsym_3| \\ \overline{\tsym_1\sim \tsym_3}}
        {\classty{X}{q_1}{\Delta_1}{\Delta_2}{\overline{\tsym_1}}
          \sim \classty{Y}{q_2}{\Delta_3}{\Delta_4}{\overline{\tsym_2}}}\and
% The following should be proved as a lemma! -Jeremy
 %\inference{\tsym_2 \sim \tsym_1}{\tsym_1\sim \tsym_2}\quad
 \inference{|\tsym_1| = |\tsym_3| \\ \overline{\tsym_1\sim \tsym_3} \\ \tsym_2 \sim \tsym_4}{\overline{\tsym_1}\to \tsym_2\sim\overline{\tsym_3}\to \tsym_4}
  \end{mathpar}
  \caption{Other metafunctions}
  \label{fig:apx:more-metafunctions}
\end{figure*}

% For example, the following Anthill program defines a class with a
% method \li{getx}, a class field \li{y}, and an instance field \li{x}. This class is then instantiated.

%   \begin{lstlisting}
% let C = class ()(*$\llparenthesis \lozenge, (\text{getx:}\emptyset\to \text{int}, \text{y:int}, (\text{x:int})\rrparenthesis$\label{ex:antint:classdef}*):
%            init=((*$\sigma$*) self. self.x=42), getx=(*$^M$*)((*$\varsigma$*) self. self.x), y=(*$^F$*)21 
% in c = C() (*\label{ex:antint:constr}*)
%   \end{lstlisting}

% When this program is translated to \upython, the following code results:

%   \begin{lstlisting}
% let C = class ():
%            init=((*$\lambda$*) self. (self(*$\Downarrow_{\text{object}(\emptyset)}$*)).x=42), getx=((*$\lambda$*) self. ((self.x)(*$\Downarrow_{\text{int}}$*)), y=21 
% in c = C()(*$\Downarrow_{\text{object}(\text{x,getx,y})}$*) (*\label{ex:muint:constr}*)
%   \end{lstlisting}

\end{document}